\documentclass[11pt]{article}

\pdfoutput=1

\usepackage{graphics} 
\usepackage{graphicx}
\usepackage{subcaption}
\usepackage{caption}
\usepackage{epsfig} 
\usepackage{times} 
\usepackage{amsmath} 
\usepackage{amssymb}  
\usepackage{amsthm}

\usepackage{enumitem}
\usepackage{lipsum}
\usepackage{mathtools}
\usepackage{cuted}
\usepackage{xcolor}
\usepackage{csquotes}

\usepackage{tikz}
\usetikzlibrary{shapes,arrows, positioning,automata}
\usepackage{float}
\usepackage{comment}

\usepackage{natbib}

\usepackage{geometry} 
\newtheorem{thm}{Theorem}
\newtheorem{exe}[thm]{Example}
\newtheorem{corol}[thm]{Corollary}
\newtheorem{ass}[thm]{Assumption}
\newtheorem{defin}[thm]{Definition}
\newtheorem{cla}[thm]{Claim}
\newtheorem{rem}[thm]{Remark}
\newtheorem{lem}[thm]{Lemma}
\newtheorem{prop}[thm]{Proposition}

\newtheorem{fct}{Fact}
\newenvironment{lemma}{\begin{lem}}{\hfill $\square$ \end{lem}}

\newenvironment{corollary}{\begin{corol}}{\hfill $\square$ \end{corol}}

\newenvironment{remark}{\begin{rem}\rm }{\hfill $\bullet$ \end{rem}}

\newenvironment{theorem}{\begin{thm}}{\hfill $\square$ \end{thm}}

\newtheorem{prob}{Problem}

\newcommand{\R}{\mathbb{R}}
\newcommand{\cA}{\mathcal A}

\newcommand{\N}{\mathbb N}
\newcommand{\G}{\mathcal{G}}

\newcommand{\x}{\times}

\makeatletter
\newcommand{\removelatexerror}{\let\@latex@error\@gobble}
\makeatother

\DeclareMathOperator{\dom}{dom}
\DeclareMathOperator{\col}{col}

\makeatletter
\def\endthebibliography{%
	\def\@noitemerr{\@latex@warning{Empty `thebibliography' environment}}%
	\endlist
}
\makeatother

\title{\LARGE \bf
A Hybrid Control Algorithm for Gradient-Free Optimization using Conjugate Directions
}

\author{Alessandro Melis\thanks{Alessandro Melis and Lorenzo Marconi are with CASY-DEI, University of Bologna, Bologna, Italy.
		{\tt\small \{alessandro.melis4, lorenzo.marconi\}@unibo.it}}, Ricardo G. Sanfelice\thanks{Ricardo G. Sanfelice is with the Department of Electrical and Computer Engineering, University of California, Santa Cruz. {\tt\small ricardo@ucsc.edu}} and Lorenzo Marconi\footnotemark[1]
\date{}
}

\begin{document}
%
%
%
%
%

\maketitle
\thispagestyle{empty}
\pagestyle{empty}

\begin{abstract}
	
	The problem of steering a particular class of $n$-dimensional continuous-time dynamical systems towards the minima of a function without gradient information is considered. We propose an hybrid controller, implementing a discrete-time Direct Search algorithm based on conjugate directions, able to solve the optimization problem for the resulting closed loop system in an almost global sense. Furthermore, we show that Direct Search algorithms based on asymptotic step size reduction are not robust to measurement noise, and, to achieve robustness, we propose a modified version by imposing a lower bound on the step size and able to achieve robust practical convergence to the optimum. In this context we show a bound relating the supremum norm of the noise signal to the step size by highlighting a trade-off between asymptotic convergence and robustness.

\end{abstract}

%

\section{Introduction}
In this paper we study the problem of steering a particular class of dynamical systems towards the minimum of an objective function, assumed to not be known but whose measurements are available at fixed intervals of time. We consider continuous-time dynamical systems that can be steered, by a known input, between any two points of the state space. Examples of such systems are completely controllable linear time-invariant systems, as well as nonlinear systems whose reachable set after time $T>0$, for all $T>0$, is the whole state space, e.g. the Dubin's vehicle (\cite{Dubin}).

The problem at hand has been tackled in the literature with a variety of approaches, mostly related to source-seeking applications. In \cite{Burian} a gradient descent method is implemented from a least-squares approximation of the gradient, and combined with an exploration phase based on a simplex algorithm, in order to steer an autonomous underwater vehicle to the deepest part of a pond, or locate hydrothermal vents. A similar problem is solved in a multi-agent framework in \cite{Bachmayer}, where, instead, local gradient measurements are assumed. In \cite{Azuma} a modified version of the simultaneous-perturbation stochastic approximation is proposed in order to recursively compute directions of exploration, and in \cite{Cochran} an extremum seeking controller is adopted assuming continuous availability of the measurements of the objective function.

In \cite{Sanf} (see also \cite{Sanf2} and \cite{Sanf3}) the source-seeking problem is solved by a hybrid controller based on the Recursive Smith-Powell (RSP) algorithm. The latter is an optimization algorithm that, through a series of line minimizations, sequentially compute a set of conjugate directions. For convex quadratic functions, it ensures to reach a neighborhood of the minimizer in a finite amount of line minimizations.

The classic RSP implementation, as in \cite{Sanf}, uses discrete line minimizations with fixed step size, able to achieve practical stability of a set of minimizers for the 2-dimensional convex quadratic case.
In \cite{Coope1999} an extension of the RSP was proposed in the general context of continuously differentiable functions. By using a decreasing step size asymptotically converging to zero, their algorithm ensures asymptotic convergence to a stationary point. While some robustness results of the RSP algorithm where shown in \cite{Sanf}, no results are present regarding the algorithm in \cite{Coope1999}, and in particular for the more general class of Direct Search methods. 

In this paper we study the class of Direct Search methods, to which the RSP algorithm belongs, which are optimization algorithms that minimize (or maximize) an objective function without using (or estimating) derivative information of any order of the objective function (see \cite{Lewis2000} for an overview). In particular we propose a direct search algorithm combining the results of \cite{Coope1999} and of \cite{Kolda2003} and \cite{Lucidi2002} in order to achieve, contrary to the RSP algorithm, asymptotic convergence to the set of minima. Due to the inherent discrete dynamics of the algorithm and the continuous dynamics of the underlying dynamical system, on the wake of \cite{Sanf}, the controller is implemented by relying on the hybrid systems framework of \cite{Teel}. The proposed hybrid controller addresses the optimization problem of an $n$-dimensional continuously differentiable function with a set of global minima, and possibly isolated local maxima, and guarantees almost global asymptotic stability of the set of minima. 

As our main focus is developing robust controllers, we also show that asymptotic Direct Search methods based on asymptotic step size reduction are in general not robust to measurement noise. Thus we propose a robust algorithm, addressing $n$-dimensional objective functions (including the results of \cite{Sanf} as a special case), highlighting that a trade-off between asymptotic convergence and robustness is mandatory.

\textit{Notation}: $\R$ denotes the set of real numbers, and $\R_{\geq 0}:=[0,\infty)$ and $\R_{\geq 1}:=[1,\infty)$. We let $e$ denote Euler's number. We denote by $|\cdot|$ the absolute value of a scalar quantity and $\|\cdot\|$ the vector 2-norm. 
For a scalar function $f:\R^n\to\R$, we denote as $\nabla f:\R^n\to\R^n$ the gradient of $f$.
Given a nonempty set $\cA\in\R^n$ and $\varepsilon>0$, we denote as $\varepsilon\mathbb{B}(\cA)$ the set $\{x\in\R^n:\|x\|_\cA<\varepsilon \}$, where $\|x\|_\cA:=\inf_{y\in\cA}\|x-y\|$. 
A set valued mapping $f$ from $\R^n$ to $\R^m$ is denoted as $f:\R^n\rightrightarrows\R^m$.
Define a hybrid system in $\R^n$ as the 4-tuple $\mathcal{H}=(C,F,D,G)$, with $C\subset\R^n$ the flow set, $F:\R^n\rightrightarrows\R^n$ the flow map, $D\subset\R^n$ the jump set, and $G:\R^n\rightrightarrows\R^n$ the jump map. Solutions to hybrid systems are defined on \textit{hybrid time domains} (see \cite{Teel} for more details) parameterized by a continuous time variable $t\in\R_{\geq 0}$ and a discrete time variable $j$, keeping track, respectively of the continuous and discrete evolution. We denote as $\dom \phi\subset \R_{\geq 0}\x\N$ the hybrid time domain corresponding to the solution $\phi$.
We say that for a hybrid system $\mathcal{H}$ with state $x\in\R^n$, the set $\mathcal{A}\subset \R^n$ is: \textit{stable} if for all $\epsilon>0$, there exists $\delta_{\epsilon}>0$ such that $x(0,0)\in\delta\mathbb{B}(\mathcal{A})$ implies $x(t,j)\in\epsilon\mathbb{B}(\mathcal{A})$ for all $(t,j)\in\dom x$; \textit{globally attractive} if $\|x(t,j)\|_{\mathcal{A}}$ is bounded and $\lim_{t+j\to \infty}\|x(t,j)\|_{\mathcal{A}}=0$, with $(t,j)\in\dom(x)$; \textit{globally asymptotically stable} if it is both stable and globally attractive; \textit{almost globally attractive} when it is globally attractive from all initial conditions apart from a set of measure zero; \textit{almost globally asymptotically stable} if it is both stable and almost globally attractive; \textit{semiglobally practically asymptotically stable} on the parameter $\theta\in\Theta\subset\R^m$, with $m>0$, if, assuming $\mathcal{H}$ complete and dependent on $\theta$, for any $\epsilon_1>\epsilon_2>0$ and there exist $\delta>0$ and $\Theta^\star\subset\Theta$ such that for all $\theta\in\Theta^\star$, $x(0,0)\in\delta\mathbb{B}(\mathcal{A})$ implies $x(t,j)\in\epsilon_1\mathbb{B}(\mathcal{A})$ and $\lim_{t+j\to \infty} \|x(t,j)\|_{\epsilon_2\mathbb{B}(\mathcal{A})}=0$ for all $(t,j)\in\dom x$.


\section{Problem Formulation}
In this paper we tackle the following optimization problem
\begin{prob}\textit{Minimize a function $f:\R^n\to\R$, namely
	\begin{equation}\label{MinPro}
	\min_{x\in\R^n} \ f(x),
	\end{equation}
	subject to the dynamics
	\begin{equation}\label{PMV}
	\dot \xi=\varphi(\xi,u) \quad\quad\quad\quad \xi=\col(x,\zeta)\in\R^{n+l},u\in\R^m.
	\end{equation}
}
\end{prob}
The state variables $x$ represent the variables involved in the optimization problem, while $\zeta$ represents other possible states. 

For simplicity we consider $\phi:\R^{n+l}\x\R^m\to\R^{n+l}$ to be continuously differentiable in $\xi$ and $u$. Moreover, given $\tau^\star>0$, we assume that for each $x_0$ and $x_f$ in $\R^n$, with $x_0\neq x_f$, there exists $t\mapsto u(t)$ such that the solution to $\dot \xi=\varphi(\xi,u(t))$ from $\xi_0=(x_0,\cdot)$, reaches $\xi_f=(x_f,\cdot)$ after $\tau^\star$ seconds. We assume that for each bounded input $\|u(t)\|\leq \bar{u}>0$ for all $t\geq 0$, $\zeta(t)$ is bounded for all $t\geq 0$. The class of systems represented by \eqref{PMV}, includes, for example, point-mass vehicles ($\xi=x$, with $x$ representing the position) and Dubin's vehicles ($\xi=\col(x^T,\zeta)$, with $x$ and $\zeta$ representing position and orientation). 

Moreover we make the following assumptions on $f$:
\begin{itemize}[leftmargin=8.5mm]
	\item[(A0)] $f$ is continuously differentiable, lower bounded and it is not assumed to be known, but sampled measurements of it are supposed to be available every $\tau^\star>0$, with $\tau^\star$ a tunable parameter;
	\item[(A1)] the set $\{x\in\R^n:\nabla f(x)=0\}$ of critical points of $f$ is such that every local minimum is also a global minimum (i.e. all local minima share the same objective function value), every local maximum is an isolated point and $f$ is analytic at every local maximum, and there are no saddle points;
	\item[(A2)] the sublevel sets of $f$, namely the sets $\mathcal{L}_f(c):=\{x\in\R^n: f(x)\leq c\}$, are compact for all $c\in\R$.
\end{itemize}

Assumptions (A0) and (A2) are standard for Direct Search methods, see \cite{Coope1999}, \cite{Kolda2003} and \cite{Lucidi2002}. Assumption (A0) can be relaxed by considering $f$ to be locally Lipschitz, as shown in \cite{Kolda2003} and \cite{NSTeel}, which requires the use of generalized gradients for analysis.

The reason for the particular structure of the set of critical points assumed in (A1) stems from the fact that our goal is to prove and guarantee convergence to the set of minima. While the assumptions on the value of the local minima is considered to simplify the structure of the problem, without the other assumptions on local maxima and saddle points, Direct Search algorithms, and our proposed controller derived from it, only guarantee convergence to the set of critical points.

Notice that, contrary to \cite{Sanf}, no convexity assumptions have been made on the cost function.

\section{The RSP and the Proposed Algorithm}
In this section we will briefly introduce the classic RSP algorithm as proposed by \cite{Smith} and \cite{Powell1964}, its hybrid implementation in \cite{Sanf} and the algorithm that we propose as an extension of the RSP.

\subsection{Background}
Throughout the paper we call \textit{line minimization} any procedure that, given a function, a direction and a point, explores the line defined by the direction applied to the point, and returns the position of the minimum, or point in a neighborhood of it, of the function along the line.

We adjective a line minimization as \textit{exact} when the minimum along the explored direction is exactly reached, and as \textit{discrete} when the line minimization is an iterative procedure that explores at each iteration a new point at distance $\Delta>0$ (fixed or changing at each iteration), called \textit{step size}, from the previously explored one. A discrete line minimization terminates when the function value of the newly explored point didn't decrease enough with respect to the function value at the last explored one.

Given a set $\G\subset\R^n$ of linearly independent directions spanning $\R^n$, the classic $RSP$ sequentially computes exact line minimizations along the directions in $\G$ in order to minimize the cost function
. Moreover, every $n$ line minimizations, a new search direction $d_{new}\in\R^n$ is computed by exploiting the $Parallel\ Subspace\ Property$ (see Theorem 4.2.1 in \cite{Roger}) and the set $\G$ is updated accordingly. 

For a convex quadratic function with Hessian matrix $H$, the newly computed direction $d_{new}$ is conjugate, by the Parallel Subspace Property, to the last $n-1$ directions in $\G$, i.e. such that $d_{new}^THd_i=0$, with $d_i\in\G$ and $i=1,...,n-1$. 

The property of conjugacy of directions for a convex quadratic function, implies that the line minimization along one direction is independent of the line minimizations along the other directions. Thus, given a set of $n$ conjugate directions for a convex quadratic function from $\R^n$ to $\R$, the minimum will be reached after $n$ line minimizations, each along a different conjugate direction. By recursively computing a set of conjugate directions,
the RSP algorithm reaches the minimum of a convex quadratic function, starting from a set of linearly independent directions, in at most $n^2$ line minimizations. This property is usually denoted as \textit{quadratic termination property}.

The version of the RSP considered in \cite{Sanf} is constrained to a 2-dimensional search space and adopts discrete line minimizations with fixed step size and an additional exploration step based on a rational rotation, whose aim is to prevent the algorithm remaining stuck for \enquote{bad} initializations. Asymptotic convergence of the algorithm to a neighborhood of the minimum, function of the step size, is proved.
\begin{figure}
	\begin{subfigure}{1\textwidth}
		\centering
		\includegraphics[scale=0.72]{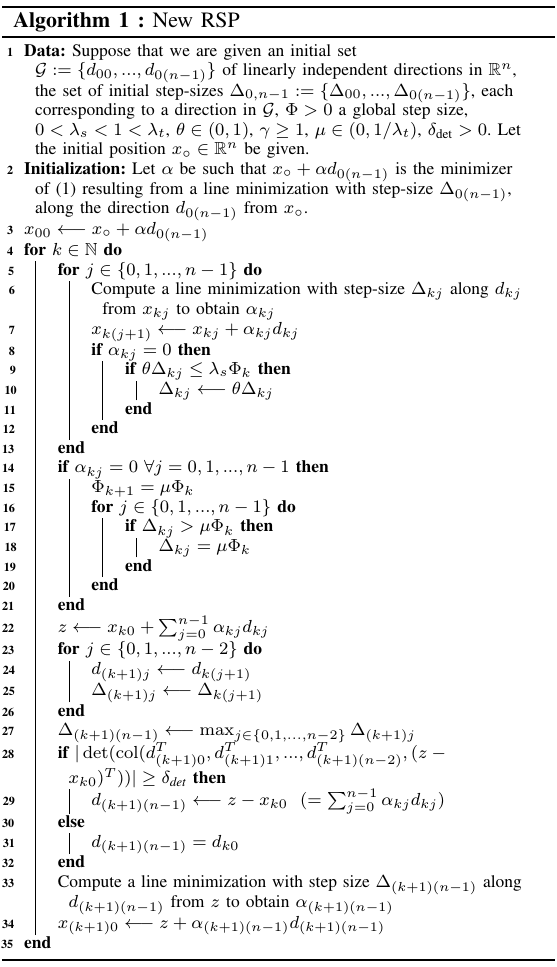}
	\end{subfigure}\\
	\begin{subfigure}{1\textwidth}
		\centering
		\includegraphics[scale=0.72]{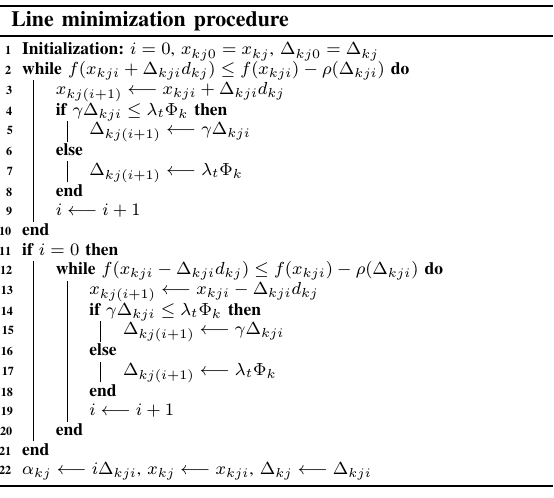}
	\end{subfigure}%

	\caption*{Alg. 1: New RSP algorithm with line minimization procedure.}
\end{figure}
\subsection{Proposed algorithm}

The algorithm proposed in this paper, shown in Alg. 1, is inspired by \cite{Garcia2002} and improves the results in \cite{Sanf} by guaranteeing, under the less restrictive assumptions (A0) and (A1), asymptotic convergence to the set of minima. 
The main differences with the RSP considered in \cite{Sanf} are reported in the following. In particular:
\begin{itemize}
	\item[1)] A different step size $\Delta_i$ is associated to each direction $d_i$ in order to guarantee more freedom of exploration. As such, when a new direction is computed (lines 28-32) also a new step size is associated to the new direction (line 27);
	\item[2)] A global step size $\Phi$ is considered, such that $\lambda_s\Phi\leq \Delta_j\leq \lambda_t\Phi$ for all $j=0,1,...n-1$. If no improvement is found along any direction, the global step size $\Phi$ is reduced to $\mu\Phi$, with $\mu\in(0,1/\lambda_t)$ (lines 14-21);
	\item[3)] In case no improvement is made along a direction (lines 8-12), the corresponding step size is reduced. This is the key step guaranteeing asymptotic convergence to the minima of the cost function;
	\item[4)] The newly computed direction is \enquote{accepted} only if it keeps the directions in $\G$ linearly independent (lines 28), otherwise the previous set of directions is retained.
\end{itemize}
\begin{remark}
	The step size associated to the newly computed direction is chosen as the maximum step size associated to the other directions, but any function bounded by the minimum and maximum of the step sizes would do. This is needed in order to guarantee that the step sizes are asymptotically reduced to zero.
\end{remark}
\begin{remark}
	The reason to reduce the step size when no improvement is found, stems from Theorem 3.3 in \cite{Kolda2003}, where it is reported that the norm of the gradient of the cost function, at points where no improvement was found along any direction, is bounded by a class $\mathcal{K}$ function of the step size. Thus, reducing the step size at those iterations, implies reducing the norm of the gradient, hence approaching a stationary point (or minimum in our case).
\end{remark}
\begin{remark}
	As pointed out in \cite{Byatt2004}, 4) is not necessary for convex quadratic functions, since every pair of conjugate directions are distanced by a minimum angle different from zero, but this is in general not true for functions satisfying assumptions (A0)-(A2).
\end{remark}

The line minimization procedure explores a direction $d_j$ from a starting point $x_{kj}$ and returns the distance $\alpha_j$ traveled from $x_{kj}$ to the found minimum of $f$ along $d_j$. The main differences in the line minimization procedure with respect to the RSP in \cite{Sanf} are the following:
\begin{itemize}
	\item[1)] Newly explored points are accepted only if a \textit{sufficient decrease} condition is satisified (lines 2 and 12), namely the function has decreased at least of $\rho(\Delta)$ along direction $d$;
	\item[2)] When a new iteration is accepted, the step size is, possibly, increased (lines 5 and 15) if the step size does not violate the upper bound imposed by the global step size.
\end{itemize}
\begin{remark}
	The sufficient decrease condition (lines 2 and 12) guarantees that the Armijo condition, needed for the algorithm convergence, is satisfied (see Section 3.7.1 in \cite{Kolda2003} for more details) and also guarantees a margin of robustness to measurement noise, as we will see in Section VI. In the sufficient decrease condition, we adopt the function $\rho:\R_{\geq 0}\to\R_{\geq 0}$ defined as
	\begin{equation}\label{rho}
	\rho(\Delta):=\begin{cases}
	\Delta^{\frac{1}{\Delta}} & \Delta\leq e\\
	\Delta +(e^{\frac{1}{e}}-e) & \Delta> e,
	\end{cases}
	\end{equation}
	and not as classically an $o(\Delta)$ function, in order to be able to escape local maxima. The function \eqref{rho} is a strictly increasing function of $\Delta$, that at $\Delta=0$ is smooth (from the right) but non-analytic, and such that $\rho(\Delta)=o(\Delta^n)$ for $\Delta\to 0$ for all $n\in\N$, implying that, under assumption (A1), if $\bar{x}\in\R^n$ is a local maxima for $f$, there exists $\bar{\Delta}>0$ such that for all $d\in\G$ and $\Delta\in(0,\bar{\Delta}]$, $f(\bar{x}+\Delta d)<f(\bar{x})-\rho(\Delta)$. Notice that any other function with the same properties as \eqref{rho} would also be an appropriate choice for $\rho$.
\end{remark}
\begin{remark}
	The step size increase during the line minimization procedure helps in better exploiting the directions in which the cost function decreases. This step does not hinder convergence of the algorithm thanks to assumption (A2).
\end{remark}
Define the set of global minima of $f$ as $\cA^\star:=\{x^\star\in\R^n:f(x^\star)\leq f(x)\ \forall x\in \R^n \}$ and as $i_{kj}^\star$ the number of steps computed in the line minimization procedure in Alg. 1 at iteration $k$ along direction $d_j$.
We can conclude the following convergence result for the algorithm in Alg. 1.
\begin{theorem}\label{Th1}
	\textit{Consider the class of cost functions fulfilling (A0)-(A2). Then, for any initial condition $x_\circ\in\R^n$ the sequence of iterate $x_{kji}$ generated by the RSP algorithm and the line minimization procedure in Alg. 1 is such that
		\begin{equation} 
		\lim_{k\to \infty} \|x_{kji}\|_{\cA^\star} = 0   \quad \forall j= 0,1,...,n-1\text{ and } \forall i\in[0,i_{kj}^\star].
		\end{equation}
	}
\end{theorem}
The proof of Theorem \ref{Th1} is based on standard arguments for the proof of convergence to stationary points of $f$ in Direct Search algorithms. In particular, under assumptions (A0) and (A2), convergence of the sequence of global step size $\Phi_k\to 0$ is shown first, which, together with the sufficient decrease condition, guarantees convergence of $x_{kji}$ to a stationary point. Under assumption (A2), and due the particular choice of \eqref{rho}, convergence to the set of minima is shown. The detailed proof of Theorem \ref{Th1} is reported in the Appendix.
\section{Hybrid Controller}
In this section we design a hybrid controller $\mathcal{H}_c$ implementing the new RSP to solve a minimization problem in $\R^n$ under the assumptions (A0)-(A2), and steer \eqref{PMV} towards the set of minima of $f$. 

The reason for resorting to the hybrid systems framework is to provide results regarding the stability and robustness of the proposed algorithm when applied to continuous-time dynamical systems, also in the presence of measurement noise. In particular, the resulting hybrid regulator is based on the framework for hybrid systems in \cite{Teel}, and its dynamics are given by a flow map $F_c$ when the state ranges in the flow set $C$, and a jump map $G_c$ when the state ranges in the jump set $D$.

The algorithm Alg.1 is implemented as a discrete time system, whose dynamics are set-valued in order to satisfy the $hybrid\ basic\ conditions$ (Assumption 6.5 in \cite{Teel}) and have the closed-loop system $\mathcal{H}_{cl}$, given by the interconnection of $\mathcal{H}_c$ and \eqref{PMV}, nominally well-posed (see Definition 6.2 in \cite{Teel}), a property needed for the application of invariance principles in the proofs of the results in the next section. 


\subsection{State of $\mathcal{H}_c$}
The state of the controller is defined as $x_c=\col(\tau,\Delta_0,...,$ $\Delta_{n-1},$ $ d_0,...,$ $d_{n-1},$ $\Phi, \lambda, \alpha^T,\bar{\alpha},p,m, $ $k,q,z$, $ \Delta,v^T)$, and it ranges in $\mathcal{X}_c:=\R_{\geq0}\times\R_{\geq0}^n\times\R^{n\times n}\times \R_{\geq 0}\times\R\times\R^n\times\R_{\geq0}\times\{-1,1\}\times\{0,1\}\times\{0,...,n\}\times\{0,1,2\}\times\R\times\R\times\R^n$. 

The state variable $\tau$ is a timer, that resets every $\tau^\star>0$ seconds, and it regulates when new cost function evaluations are available.

Its hybrid dynamics are given by
\begin{equation}\label{Tim1}
\dot{\tau} =1 \quad \quad (\xi,x_c)\in C:=\{(\xi,x_c)\in\R^{n+l}\times\mathcal{X}_c:\tau\leq\tau^\star\},
\end{equation}
during flow, and
\begin{equation}\label{Tim2}
\tau^+=0 \quad \quad (\xi,x_c)\in D:=\{(\xi,x_c)\in\R^{n+l}\times\mathcal{X}_c:\tau\geq\tau^\star\},
\end{equation}
during jumps. 

The states $d_j\in\R^n$ and $\Delta_j\in\R_{\geq0}$, $j=0,1,...,n$ represent, in Alg. 1, the search directions and the step sizes corresponding to each direction.
The state variable $\lambda\in\R$, which keeps track of the distance traveled along the currently explored direction, and the state variable $\alpha\in\R^n$, which stores the total traveled vector from direction $d_0$, are related to the distance traveled along each direction, which is the variable $\alpha_j$ introduced in Alg. 1.

The state $\Phi\in\R_{\geq 0}$ represents the global step size and $\bar{\alpha}\in\R_{\geq 0}$ the total distance traveled during each cycle of directions exploration.

The positive or negative exploration along the current direction is determined by the state $p\in\{-1,1\}$, and the variable $m\in\{0,1\}$ indicates whether a turn has already been computed along the current direction.

To define in which operating point of the modified RSP algorithm the controller is, the state variables $k\in\{0,1,...,n\}$ and $q\in\{0,1,2\}$ have been introduced. The variable $k$ represents the state of the RSP, namely which direction is currently being explored. Notice that it has $n+1$ components since the direction $d_{n-1}$ is explored twice to be able to exploit the $Parallel\ Subspace\ Property$. The variable $q$, defining the state of the line minimization, assumes these values
\begin{itemize}
	\item[-] $q=0$: the positive line minimization;
	\item[-] $q=1$: the negative line minimization;
	\item[-] $q=2$: the line minimization is completed.  
\end{itemize}

The state variable $z\in\R$ is a memory state that keeps track of the best minimum value of $f$ found satisfying the sufficient decrease condition. 

Two more states have been added for ease of notation, namely $\Delta\in\R$ and $v\in\R^n$, that store the currently explored search direction and its corresponding step size.

\subsection{Hybrid Controller Structure}

The structure of $\mathcal{H}_c$ is given by
\begin{equation}\label{Hc}
\mathcal{H}_c\ : \ \begin{cases}\dot{x}_c = F_c:= \begin{bmatrix}
1 &
0 &
\hdots &
0
\end{bmatrix}^T & \hspace*{-3mm}(x,x_c)\in C\\
x_c^+\in G_c(x_c, f(x)):=\begin{bmatrix}
0\\
G_{c/\tau}(x_c,f(x))
\end{bmatrix} & \hspace*{-3mm}(x,x_c)\in D\\
u=K(x,x_c,\tau^\star),
\end{cases}
\end{equation}
with sets $C,\ D$ defined before. The flow map $F_c$ is a single-valued constant function with all components equal to zero except for the timer. 
The jump map $G_c:\mathcal{X}_c\times\R\to\mathcal{X}_c$ is a set-valued map, composed by the timer discrete dynamics and $G_{c/\tau}:\mathcal{X}_c/\R\times\R\to\mathcal{X}_c/\R$, representing the hybrid implementation of Alg. 1. The output of $\mathcal{H}_c$ is a function $K:\R^n\x\mathcal{X}_c\x\R_{> 0}\to\R^m$ that steers the $x$-subsystem from $x(t_j,j)$ to $x(t_j+\tau^\star,j) =x(t_j,j)+p(t_j,j)\Delta(t_j,j) v(t_j,j)$, with $t_j=\inf_{t\in\R_{\geq 0}}(t,j)\in\dom x$, for all $j\in\N$.

The set-valued map $G_{c/\tau}$ is reported in the Appendix. 

We stress that, as in the current implementation the step size is reduced and the timer limit is kept constant, the speed of system \eqref{PMV} is reduced proportionally by reduction of the step size. In this way the distance traveled during the flow gets smaller and smaller, and the state $x$ asymptotically converges to the set of minima. 
\section{Stability Analysis}
Define the hybrid closed-loop $\mathcal{H}_{cl}$ as the interconnection of the dynamics \eqref{PMV} and the controller $\mathcal{H}_c$ developed in the previous section, namely
\begin{equation}\label{HCL}
\mathcal{H}_{cl}\ : \ \begin{cases}\begin{aligned}
&\dot \xi = \varphi(\xi,K(x,x_c,\tau^\star))  \\
&\dot{x}_c = F_c  \\
&\xi^+=\xi \\
&x_c^+\in G_c(x_c, f(x)) \\
\end{aligned}
\begin{aligned}
&\left.\vphantom{\begin{aligned}
	&\dot x = K(x_c)  \\
	&\dot{x}_c = F_c 
	\end{aligned}}\right\rbrace\quad(\xi,x_c)\in C\\
&\left.\vphantom{\begin{aligned}
	&x^+=x\\
	&x_c^+\in G_c(x_c, f(x))
	\end{aligned}}\right\rbrace\quad(\xi,x_c)\in D
\end{aligned}
\end{cases}
\end{equation}
The flow and jump maps of the closed-loop system $\mathcal{H}_{cl}$ are thus defined as $F(\xi,x_c):=$ $\col(\varphi(\xi,$ $K(x,x_c,\tau^\star)),F)$ for all $(\xi,x_c)\in C$ and $G(\xi,x_c):=\col(\xi,G(x_c,f(x)))$ for all $(\xi,x_c)\in D$.

Define $\cA_{dis}:=\{-1,1\}\times\{0,1\}\times\{0,1,...,n\}\times \{0,1,2\}$. We consider the stabilization problem with respect to the sets $\mathcal{A}\subset \mathcal{A}_e\subset\R^{n+l}\x\mathcal{X}_c$, defined as
\begin{align}\label{A}
\begin{array}{ll}
\mathcal{A}:=\R^l\times\cA^\star&\times[0,1]\times\{0^n\}\times\R^{n\x n}\times\{0\}\times\{0\}\times\\&\{0^n\}\times\{0\}\x\cA_{dis}\times\{f(\cA^\star)\}\times\{0\}\times\R^n,
\end{array}
\end{align}
\begin{align}\label{A_e}
\begin{array}{ll}
\mathcal{A}_e:=&\R^{n+l}\times[0,1]\times\{\{0^n\}\times\R^{n\x n}\times\{0\}\cup\R^n\x\\&\{0^{n\x n}\}\x\R_{\geq 0}\}\times\{0\}\times\{0^n\}\times\{0\}\times\cA_{dis}\\&\x\R\times\{\{0\}\times\R^n\cup\R\x\{0^n\}\}.
\end{array}
\end{align}
The set $\cA$ represents the desired equilibrium set, namely the subset of $\R^{n+l}\x\mathcal{X}_c$ such that if $(\xi(0,0),x_c(0,0))\in\cA$, then $x(t,j)\in\cA^\star$ for all $(t,j)\in\dom (\xi,x_c)$. Notice that invariance of $\cA$ is guaranteed by all the step size variables being zero, so that $x(t_j+\tau^\star,j) =x(t_j,j)$. However, $x(t_j+\tau^\star,j) = x(t_j,j)$, namely no optimization step is computed, also for an initialization with $\Phi(0,0)=0$ and/or $d_j(0,0)=0$, even if $x(0,0)\notin\cA^\star$. The set $\cA_e$ takes into account this consideration, indeed it is the set of equilibrium points for $\mathcal{H}_{cl}$ for which no optimization step is performed due to an initialization with $\Phi=0$ and/or $d_j=0$ for all $j=0,1,...,n-1$.

\begin{theorem}\label{Thm2}\textit{
		Let assumptions (A0)-(A2) hold, $\tau^\star>0$, and the parameters of the algorithm Alg. 1 satisfy $\delta_{det}>0$, $0<\lambda_s<1<\lambda_t$, $\mu\in(0,1/\lambda_t)$, $\theta\in(0,1)$ and $\gamma\geq1$. Then, for the closed-loop system $\mathcal{H}_{cl}$, the set $\mathcal{A}$ in \eqref{A} is
		\begin{itemize}
			\item stable;
			\item  almost globally attractive;
		\end{itemize}
		hence it is almost globally asymptotically stable. Furthermore, the set $\cA_e$ in \eqref{A_e} is globally attractive for $\mathcal{H}_{cl}$.
	}	
\end{theorem}

The proof of Theorem \ref{Thm2} and of the results in the next section, based on Lyapunov arguments and invariance principles, applied considering the Lyapunov candidate function $V(\xi,x_c):=z-f(\cA^\star)$, are included in the Appendix. 

\begin{remark}
	From Theorem \ref{Thm2} and the structure of $\cA$ and $\cA_e$, it follows in particular that, for any initialization such that det$(\col(d_0,d_1,...,d_{n-1}))\neq0$ and $\Phi\neq0$, boundedness of the closed-loop trajectories and asymptotic convergence to the set $\cA$ are guaranteed.
\end{remark}

\begin{remark}
	Notice that, depending on the values of the constants $\delta_{det}$, the quadratic termination property can be lost. Nonetheless, the asymptotic convergence property is preserved. 
\end{remark}

\section{About Robustness of the Algorithm}
In this section we investigate the robustness of the proposed algorithm to noise acting on the cost function measurements. We start with a negative result showing that general Direct Search Algorithms based on line minimizations and asymptotic step size reduction are not robust to any bounded measurement noise.
\begin{theorem}\label{TH3}
	\textit{Consider the class of Direct Search algorithms based on line minimizations and with asymptotic step size reduction, to which the algorithm Alg. 1 belongs to, acting on a function $f:\R^n\to\R$ satisfying assumptions (A0) and (A2). Then, for any bound $\bar{n}_s>0$, there exists a noise $n_s:\R\to\R$, with $|n_s(t)|\leq \bar{n}_s$ $\forall t\in\R$, such that, for noisy cost function measurements, namely $f(x(t))+n_s(t)$, and all initial conditions apart from a set of measure zero, the sequence of iterate produced by such algorithms escapes any compact sub-level set of $f$.}
\end{theorem}

The above result shows that there is no robustness guarantee for the modified RSP algorithm, even if stability has been shown and convergence results are attainable for a proper choice of initial conditions.

Robustness to measurement noise for the hybrid closed-loop system $\mathcal{H}_{cl}$ is recovered by imposing a lower bound $\underline{\Phi}>0$ on the global step size $\Phi$, and modifying accordingly $G_{c\setminus\tau}$. 
In particular, in $g_5$, the discrete dynamics of $\Phi$ can be modified as follows.
\begin{equation}\label{NewPhi}
\Phi^+=\begin{cases}
\mu\Phi & \text{if }k=n\text{ and }\bar{\alpha}+\|\lambda v\|\leq \min\limits_{j\in\{0,1,...,n-1\}}\Delta_{j}/2\text{ and }\mu\Phi\geq \underline{\Phi}\\
\underline{\Phi} & \text{if }k=n\text{ and }\bar{\alpha}+\|\lambda v\|\leq \min\limits_{j\in\{0,1,...,n-1\}}\Delta_{j}/2\text{ and }\mu\Phi\leq \underline{\Phi}\\
\Phi & \text{otherwise}.
\end{cases}
\end{equation}
Moreover, given $\delta_{\text{det}}>0$, we restrict the domain of all the directions $d_j$ to be such that $\det(\col(d_0,d_1,...,d_{n-1}))\geq \delta_{\text{det}}$. Without loss of generality, we will denote the desired equilibrium set within the restricted domain for the directions as $\cA$. 

\begin{theorem}\label{TH4}
	\textit{Let assumptions (A0)-(A2) hold, $\underline{\Phi}>0$, the parameters of the algorithm Alg. 1 satisfy  $0<\lambda_s<1<\lambda_t$, $\delta_{\text{det}}>0$, $\mu\in(0,1/\lambda_t)$, $\theta\in(0,1)$ and $\gamma\geq1$, with the update of $\Phi$ modified such that $\Phi(t,j)\geq \underline{\Phi}$ for all $(t,j)\in \dom \Phi$. Then the set $\cA$ is semiglobally practically asymptotically stable on $\underline{\Phi}>0$ for $\mathcal{H}_{cl}$.}
\end{theorem}

The lower bound on $\Phi$ also guarantees an explicit bound on the allowable maximum noise that can be accepted without losing robustness. 

\begin{corollary}\label{Cor1}
	\textit{For all parameters of the algorithm Alg. 1 satisfying $0<\lambda_s<1<\lambda_t$, $\delta_{\textnormal{det}}>0$, $\mu\in(0,1/\lambda_t)$, $\theta\in(0,1)$, $\gamma\geq1$, and all measurement noise $n_s:\R\times\mathbb{N}\to\R$ added to $f$, with $|n_s(t,j)|\leq \bar{n}_s$ for all $(t,j)\in\R\times\mathbb{N}$, with $\bar{n}_s>0$, pick $\underline{\Phi}^\star>0$ such that
		\begin{equation}\label{Bound}
		\bar{n}_s= \frac{\rho(\lambda_s\underline{\Phi}^\star)}{2}.
		\end{equation} 
		Then the set $\cA$ is semiglobally practically asymptotically stable on $\underline{\Phi}\geq\underline{\Phi}^\star$ for $\mathcal{H}_{cl}$, with the update of $\Phi$ modified such that $\Phi(t,j)\geq \underline{\Phi}$ for all $(t,j)\in \dom \Phi$. 
	}
\end{corollary}

\begin{remark}
	In \cite{Sanf} an explicit characterization of the practical neighborhood of convergence to $\cA$, as function of the step size, is provided. As the dense exploration procedure adopted in \cite{Sanf} to guarantee such bounds cannot be extended to $n$-dimensional search spaces, a similar result cannot be achieved without further assumptions on $f$. Nonetheless, the norm of the gradient of $f$ can be bounded at steady state by a function of $\underline{\Phi}$ and the equilibrium set of exploring directions (see Theorem 3.3 in \cite{Kolda2003}).
\end{remark}
\begin{remark}
	The trade-off between practical global asymptotic stability and almost global asymptotic stability is, also, related to the lack of knowledge of $\cA^\star$ or $f(\cA^\star)$. By assuming, for example, knowledge of $f(\cA^\star)$, the discrete dynamics of $\Phi$ can be extended with the addition of a term $\rho_f(|f(x)-f(\cA^\star)|)$, where $\rho_f:\R_{\geq0}\to\R_{\geq 0}$ and $\rho_f(|f(x)-f(\cA^\star)|)>0$ for $x$ such that $|f(x)-f(\cA^\star)|>0$. This term would prevent the algorithm to remain stuck at the initial position when $\Phi$ is initialized at zero and, thus, Theorem \ref{Thm2} could be extended to guarantee global asymptotic stability of the set of minimizers.
\end{remark}
\section{Simulations Results}
In this section we show the results of different simulations of the proposed hybrid controller to the minimization of different objective functions.

Fig. 1 illustrates the level sets of the quadratic convex function
\begin{equation}\label{quad}
f(x)=x_1^2+5x_2^2,
\end{equation}
where $x=\col(x_1,x_2)$. The trajectory of a point-mass vehicle, steered by the proposed hybrid controller in order to minimize \eqref{quad}, is superimposed to the level sets of \eqref{quad}, showing the value of $f(x)$ at each corresponding point of the trajectory. The control input was chosen as $K(x,x_c,\tau^\star)=p\Delta v/\tau^\star $. For this simulation, the initial values of the state variables of the hybrid closed loop were chosen as $x(0,0)=\col(1.5,0)$, $ \tau(0,0)=0,\ \lambda(0,0)=0,\ \alpha(0,0)=0,\ z(0,0)=0,\ p(0,0)=1,\ q(0,0)=0,\ m(0,0)=0,\ k(0,0)=0,\ \alpha(0,0)=0,\ v(0,0)=\col(\cos(\pi/8), \sin(\pi/8)),\ \Delta(0,0)=0.01,\ d_0(0,0)=v,\ d_1=\col(-\sin(\pi/8),$ $ \cos(\pi/8)),\ \Phi(0,0) = 0.01, \ \Delta_j(0,0)=\Delta(0,0),\ j=0,1$. The tunable parameters of the controller were defined as $\gamma=1.2,\ \theta=0.5,\ \delta_{\textnormal{det}}=0.001, \ \mu = 0.15, \ \lambda_s= 0.001, \ \lambda_t= 5$.

\addtocounter{figure}{-1}
\begin{figure}[htp]
	\centering
	\begin{subfigure}{.5\linewidth}
		\includegraphics[scale=0.5]{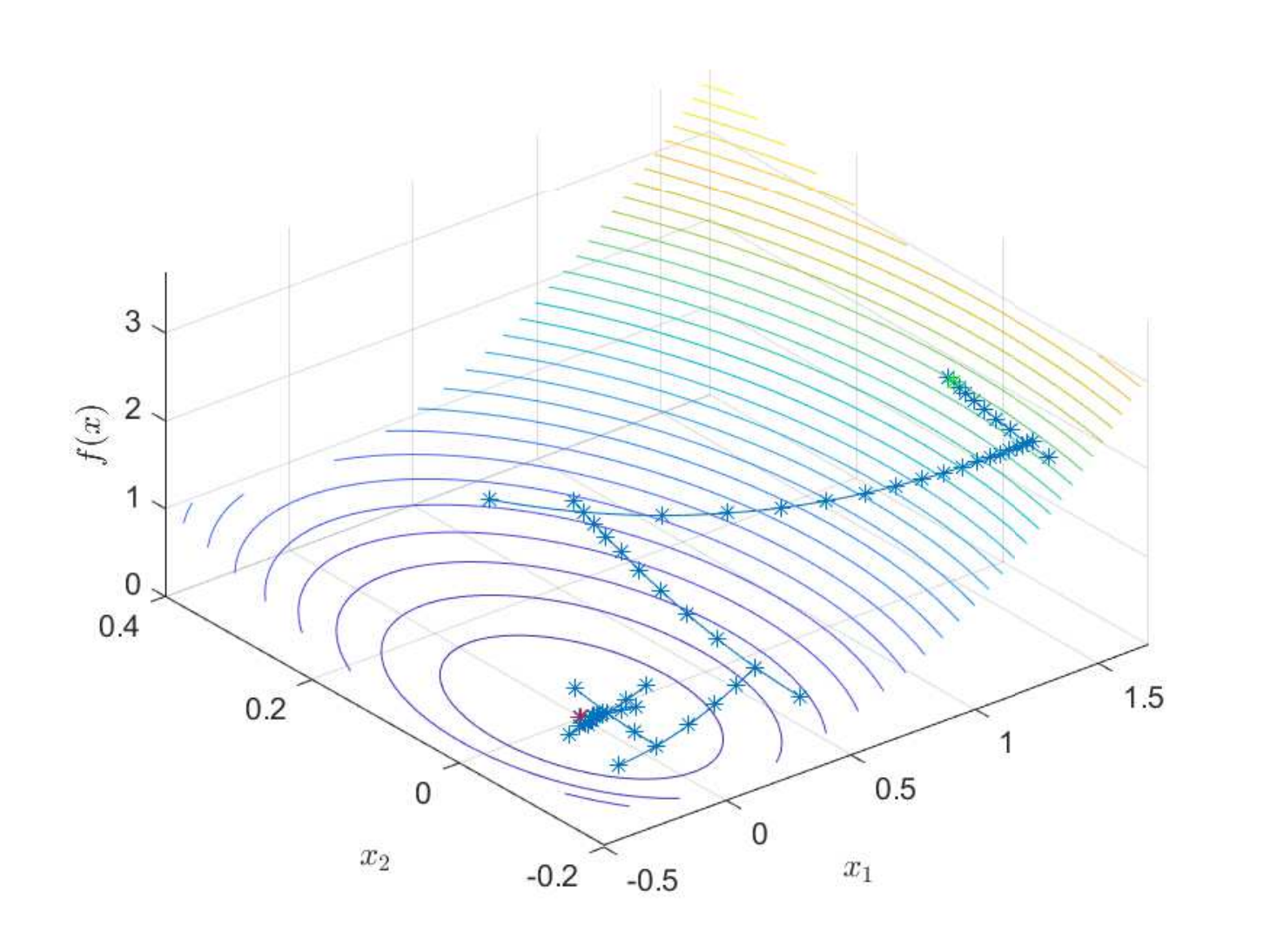}
		\centering
		\caption{ $x$ trajectory versus the level sets of a quadratic convex function}
	\end{subfigure}%
	\begin{subfigure}{.5\linewidth}
		\includegraphics[scale=0.5]{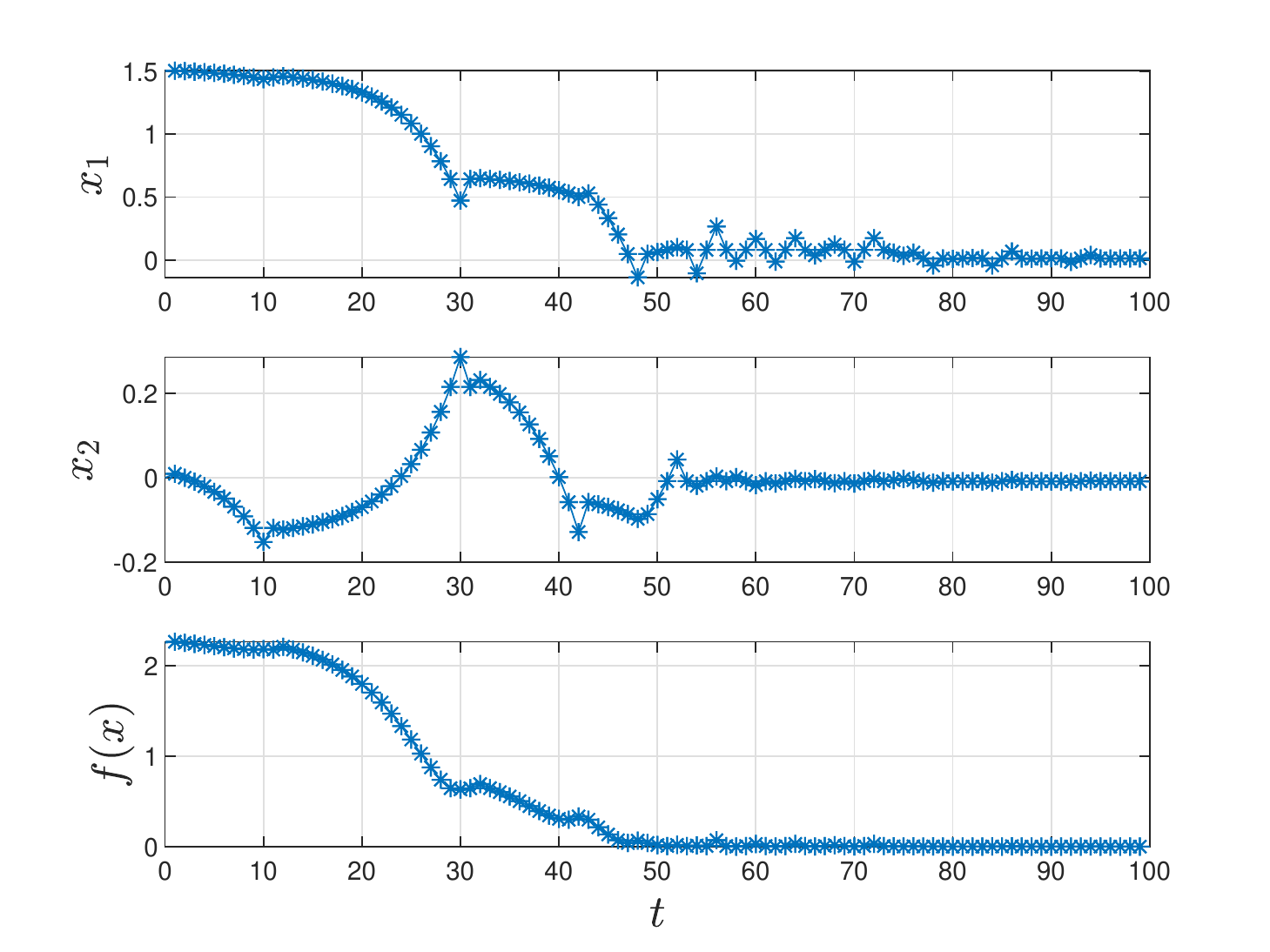}
		\centering
		\caption{$x(t)$ and $f(x(t))$}
	\end{subfigure}
	\caption{Plot of the trajectories of $x(t,j)$ and $f(x(t,j))$, where $f(x)=x_1^2+5x_2^2$. (a) Shows the vehicle path (blue with '*' where jump occurs) on the level sets of $f$. The initial point is indicated with a green '*' and the unique minimizer $(0,0,0)$ with a red '*'. (b) Shows the evolution of $x$ and $f(x)$ as function of time.}
\end{figure}
\begin{figure}[htp]
	\includegraphics[scale=0.4]{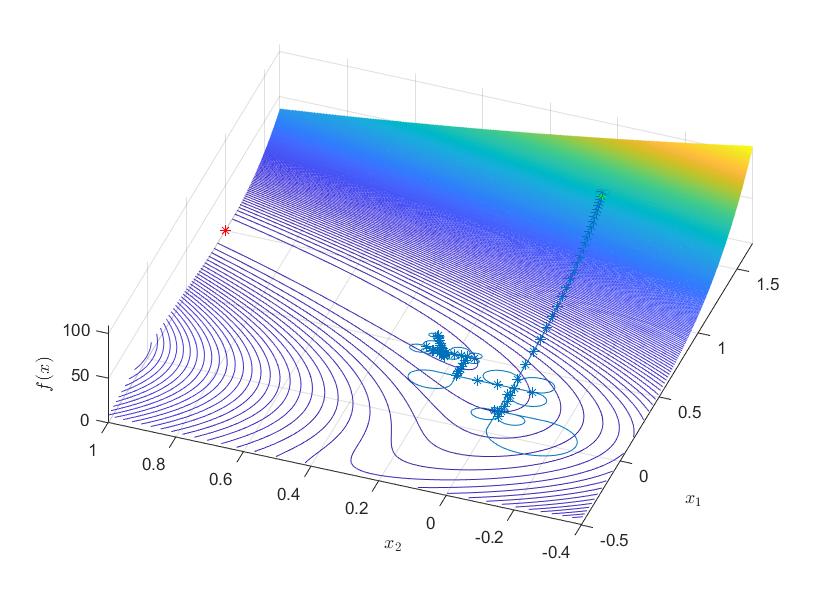}
	\centering
	\caption{The Dubin's vehicle path on the level sets of the Rosenbrock function \eqref{Ros}.}
\end{figure}

It can be noticed as in both Fig. 1(a) and Fig.1(b), the distance to the minimizer tends asymptotically to zero as the step size converges to zero.
The simulation reported in Fig. 2, instead, considered the nonconvex Rosenbrock function
\begin{equation}\label{Ros}
f(x)=(1-x_1)^2+10(x_2-x_1^2)^2,
\end{equation}
and the Dubin's vehicle dynamics
\begin{align*}
	\dot x_1 &= V\cos(\zeta)\\
	\dot x_2 &= V\sin(\zeta)\\
	\dot \zeta &= u,
\end{align*}
with $(x_1,x_2)\in\R^2$ the position, $V>0$ the velocity constant, $\zeta\in\R$ the orientation, and $u\in\R$ the control input. The initial conditions and parameter values were kept the same of the previous simulation. In this case the minimizer is given by $x^\star=(1,1)$ and, in spite of the nonconvex optimization problem, the trajectory of the state variable $x$ is converging towards it, remarkably.

\section{Conclusion}
This paper presents an extension of the results in \cite{Sanf}. In particular, an hybrid controller based on a modified RSP algorithm, which optimizes an objective function without gradient information, and that is able to achieve almost global asymptotic stability of the closed loop composed by the controller and a particular class of continuous-time dynamical systems is proposed. As direct search methods are shown to not be robust to measurement noise, a modified practical scheme is proposed, a bound relating the minimum allowable step size and the measurement noise supremum norm is reported, and stating how a trade-off between asymptotic convergence and robustness is inevitable for this class of algorithms. Simulations results are provided to validate the proposed approach. Future developments include the extension of the proposed controller to the multiagent scenario, in order to efficiently exploit the parallel subspace property, and to more general objective functions, e.g. to nonsmooth functions. 

\appendix
\setcounter{secnumdepth}{0}
\section{Appendix}

\subsection{The map $G_{c\setminus\tau}$}
The set-valued map $G_{c/\tau}$ is presented next. It is given by the composition of the maps $g_i(x_c,f(x_c))$ defined on the subsets $D_i$, $i=1,2,...,5$ of the jump set $D$, namely $G_{c/\tau}(x_c,f(x)) := g_i(x_c,f(x))$ for $(x_c,x)\in D_i$, where $D=\cup_{i=1}^{5}D_i$. We omit the update law of the state variables that remain constant at jumps.

The sets $D_i$ define the conditions under which the different operations of the algorithm proposed, integrated in the functions $g_i$, take place. 

\begin{itemize}[leftmargin=6mm]
	\item[1)] Continue a positive line search:\\
	$D_1=\{ (x,x_c) \in \R^n\times\mathcal{X}_c :\ f(x)\leq z -\rho(\Delta)$, $p=1$, $q\in\{0,1\}$, $m=0 \}$\\
	$g_1:\ z^+=f(x)$, $q^+=1$, \\	
	$\lambda^+=\begin{cases}
	\lambda+\Delta_{n-1}p & \quad \quad \quad \text{if }k=0\\
	\lambda+\Delta_{k-1}p & \quad \quad \quad \textnormal{otherwise}
	\end{cases}
	$\\
	$	\Delta_{k-1}^+=\begin{cases}
	\gamma\Delta_{k-1} & \text{if }k=1,...,n-1 \text{ and }\gamma\Delta_{k-1}\leq \lambda_t\Phi\\
	\lambda_t\Phi  & \text{if }k=1,...,n-1\text{ and }\gamma\Delta_{k-1}\geq \lambda_t\Phi
	\end{cases}
	$\\
	$	\Delta^+_{n-1}=\begin{cases}
	\gamma\Delta_{n-1} & \text{if }k=0,n \text{ and }\gamma\Delta_{n-1}\leq \lambda_t\Phi\\
	\lambda_t\Phi  & \text{if }k=0,n\text{ and }\gamma\Delta_{n-1}\geq \lambda_t\Phi
	\end{cases}
	$\\
	$\Delta^+=\begin{cases}
	\gamma\Delta_{k-1} & \text{if }k=1,2,...,n-1 
	\text{ and }\gamma\Delta\leq \lambda_t\Phi\\
	\gamma\Delta_{n-1} &  \text{if }k=0,n\text{ and }\gamma\Delta\leq \lambda_t\Phi\\
	\lambda_t\Phi & \text{if }\gamma\Delta\geq \lambda_t\Phi
	\end{cases}
	$\\
	\item[2)] Correct overshoot:\\
	$D_2=\{(x,x_c) \in \R^n\times\mathcal{X}_c : f(x)\geq z -\rho(\Delta)$, $q\in\{0,1\}$, $m=0 \}$\\
	$g_2:\ p^+=-p$, $m^+=1$, $q^+=q+1$,
	%
	%
	%
	\item[3)] Starting negative line search:\\
	$D_3=\{(x,x_c) \in \R^n\times\mathcal{X}_c :m=1$, $p=-1$, $q=1 \}$\\
	$g_3:\ z^+=f(x)$, $m^+=0$, $\lambda^+=0$,
	\item[4)] Continue a negative line search:\\
	$D_4=\{(x,x_c) \in \R^n\times\mathcal{X}_c :f(x)\leq z-\rho(\Delta)$, $p=-1$, $q=1$, $m=0\}$\\
	$g_4:\ z^+=f(x),$ \\
	$\lambda^+=\begin{cases}
	\lambda+\Delta_{n-1}p &\quad \quad \quad \text{if } k=0\\
	\lambda+\Delta_{k-1}p &\quad \quad \quad \text{otherwise}
	\end{cases}$\\
	$	\Delta_{k-1}^+\begin{cases}
	\gamma\Delta_{k-1} & \text{if }k=1,...,n-1 \text{ and }\gamma\Delta_{k-1}\leq \lambda_t\Phi\\
	\lambda_t\Phi  & \text{if }k=1,...,n-1\text{ and }\gamma\Delta_{k-1}\geq \lambda_t\Phi
	\end{cases}
	$\\
	$	\Delta^+_{n-1}=
	\begin{cases}
	\gamma\Delta_{n-1} & \text{if }k=0,n \text{ and }\gamma\Delta_{n-1}\leq \lambda_t\Phi\\
	\lambda_t\Phi  & \text{if }k=0,n\text{ and }\gamma\Delta_{n-1}\geq \lambda_t\Phi
	\end{cases}
	$\\
	$\Delta^+=\begin{cases}
	\gamma\Delta_{k-1} &    \text{if }k=1,2,...,n-1\text{ and }\gamma\Delta\leq \lambda_t\Phi\\
	\gamma\Delta_{n-1} &    \text{if }k=0,n\text{ and }\gamma\Delta\leq \lambda_t\Phi\\
	\lambda_t\Phi &    \text{if }\gamma\Delta\geq \lambda_t\Phi
	\end{cases}
	$\\
	\item[5)] Update direction and start positive line search:\\
	$D_5=\{(x,x_c) \in \R^n\times\mathcal{X}_c : q=2\}$\\
	$g_5:$\\
	$q^+=0,\ p^+=1,\ \lambda^+=0,\ m^+=0, z^+=f(x)$\\ 
	$\alpha^+=\begin{cases}
	\alpha+\lambda v &\quad \quad \quad \text{if } k=0,1,...,n-1 \\
	0 & \quad \quad \quad\text{if } k=n
	\end{cases}$\\
	$\bar{\alpha}^+=\begin{cases}
	\bar{\alpha}+\|\lambda v\| &\quad \quad \quad \text{if } k=0,1,...,n-1 \\
	0 & \quad \quad \quad\text{if } k=n
	\end{cases}
	$\\
	$k^+=(k+1)\mod n+1$\\
	$\Delta^+=\begin{cases}
	\Delta_{(k\mod n+1)} & \textnormal{if }k=0,1,...,n-1\\
	\mu\lambda_t\Phi & \text{if }k=n \text{ and }\bar{\alpha}+\|\lambda v\|\leq \min\limits_{j\in\{0,1,...,n-1\}}\frac{\Delta_{j}}{2}\\
	\lambda_t\Phi & \text{otherwise}
	\end{cases}$\\
	$\Phi^+=\begin{cases}
	\mu\Phi & \text{if }k=n\text{ and }\bar{\alpha}+\|\lambda v\|\leq \min\limits_{j\in\{0,1,...,n-1\}}\frac{\Delta_{j}}{2}\\
	\Phi & \text{otherwise}
	\end{cases}
	$\\
	$v^+=\begin{cases}
	\phi(\alpha,\lambda v, M_{1,n-1},d_0) & \text{if }k=n\\
	d_k & \text{otherwise}
	\end{cases}
	$\\
	$d_0^+=\begin{cases}
	d_0 & \text{if }k=0,...,n-1\\
	d_1 & \text{if }k=n
	\end{cases}
	$\\
	$
	\hspace*{0.3cm}
	\vdots$\\
	$d_{n-2}^+=\begin{cases}
	d_{n-2} & \text{if }k=0,...,n-1\\
	d_{n-1} & \text{if }k=n
	\end{cases}
	$\\
	$d_{n-1}^+=\begin{cases}
	d_{n-1} & \text{if }k=0,1,...,n-1\\
	\phi(\alpha,\lambda v, M_{1,n-1},d_0) & \text{if }k=n
	\end{cases}
	$\\
	$\Delta_{0}^+=\begin{cases}
	\theta\Delta_{0} & \text{if }k=1\text{ and }|\lambda|\leq\frac{\Delta_0}{2}\text{ and }\theta\Delta_{0}\geq\lambda_s\Phi\\
	\lambda_s\Phi & \text{if }k=1\text{ and }|\lambda|\leq\frac{\Delta_0}{2}\text{ and }\theta\Delta_{0}\leq\lambda_s\Phi\\
	\Delta_0 & \text{if }k=0,2,...,n-1 \text{ or }(k=1\text{ and }|\lambda|\geq\frac{\Delta_0}{2})\\
	\mu\lambda_t\Phi & \text{if }k=n\text{ and }\Delta_1\geq\mu\lambda_t\Phi\text{ and }\bar{\alpha}+\|\lambda v\|\leq \min\limits_{j\in\{0,1,...,n-1\}}\frac{\Delta_{j}}{2}\\
	\Delta_1 & \text{otherwise}\\
	\end{cases}
	$\\
	$\hspace*{0.3cm}\vdots$\\
	$\Delta_{n-2}^+= \begin{cases}
	\theta\Delta_{n-2} & \text{if }k=n-1\text{ and }|\lambda|\leq\frac{\Delta_{n-2}}{2}\text{ and }\theta\Delta_{n-2}\geq\lambda_s\Phi\\
	\lambda_s\Phi & \text{if }((k=n-1\text{ and }|\lambda|\leq\frac{\Delta_{n-2}}{2})\text{or }(k=n\text{ and }|\lambda|\leq\frac{\Delta_{n-1}}{2}))\\&\text{and }\theta\Delta_{n-2}\leq\lambda_s\Phi\\
	\Delta_{n-2} & \text{if }k=0,...,n-2  \text{ or }(k=n-1\text{ and }|\lambda|\geq\frac{\Delta_{n-2}}{2})\\
	\theta\Delta_{n-1} & \text{if }k=n\text{ and }|\lambda|\leq\frac{\Delta_{n-1}}{2}\text{ and }\Delta_{n-1}\geq\mu\lambda_t\Phi\\
	\mu\lambda_t\Phi & \text{if }k=n\text{ and }\Delta_{n-1}\geq\mu\lambda_t\Phi\text{ and }\bar{\alpha}+\|\lambda\|\leq \min\limits_{j\in\{0,...,n-1\}}\frac{\Delta_{j}}{2}\\
	\Delta_{n-1} & \text{otherwise}
	\end{cases}
	$\\
	$\Delta_{n-1}^+= \begin{cases}
	\theta\Delta_{n-1} & \text{if }k=0\text{ and }|\lambda|\leq\frac{\Delta_{n-1}}{2}\text{ and }\theta\Delta_{n-1}\geq\lambda_s\Phi\\
	\lambda_s\Phi & \text{if }k=0\text{ and }|\lambda|\leq\frac{\Delta_{n-1}}{2}\text{ and }\theta\Delta_{n-2}\leq\lambda_s\Phi\\
	\Delta_{n-1} & \text{if }k=1,...,n-1 \text{ or }
	(k=0\text{ and }|\lambda|\geq\frac{\Delta_{n-1}}{2})\\
	\mu\lambda_t\Phi & \text{if }k=n\text{ and }\bar{\alpha}+\|\lambda\|\leq \min\limits_{j\in\{0,1,...,n-1\}}\frac{\Delta_{j}}{2}\\&\text{and }\max\limits_{j\in\{0,...,n-2\}}\Delta_{j}\geq\mu\lambda_t\Phi\\
	\max\limits_{j\in\{0,...,n-2\}}\Delta_{j} &\text{otherwise}
	\end{cases}
	$\\
\end{itemize}

Even if the logic of the jump map $G_{c/\tau}$ resembles the one in the new RSP in Alg. 1, a couple of explanations are in order.

The computation of the new conjugate direction in $g_5$ is addressed by the function $\phi:\R^n\times\R^n\times\R^{n\times (n-1)}\times\R^n$ defined as
\begin{align*}
&\phi(\alpha,\beta, M_{1,n-1},d_0)=\\
&\begin{cases}
\alpha+\beta &  \det(\col(M_{1,n-1}^T,(\alpha+\beta)^T)^T)>\delta_{\text{det}}\\
d_0 & \det(\col(M_{1,n-1}^T,(\alpha+\beta)^T)^T)<\delta_{\text{det}}\\
\left\{d_0, \alpha+\beta \right\} & \text{otherwise},
\end{cases}
\end{align*}
where $M_{1, n-1}:=\col(d_1^T,...,d_{n-1}^T)^T$.
The conditions in $\phi$ check if the new direction $\alpha+\beta$, that is going to be computed exploiting the $Parallel\ Subspace\ Property$, is linearly independent from the last $n-1$ directions, namely if the determinant of the concatenation of $M_{1, n-1}$ and the new direction is bigger than a tunable parameter $\delta_{det}>0$. In case this condition is not satisfied, the previous set of directions is retained.

%

The update rule of the states $\Delta_j$, $j=0,1,...,n-1$ also needs clarification. Let us consider $\Delta_{n-1}^+$ since the same reasoning applies to the other state variables. The condition $	|\lambda|<\Delta_{n-1}/2$ is a different way to express the condition $\lambda=0$, while at the same time satisfying outer semicontinuity of the map $g_5$. Indeed $|\lambda|<\Delta_{n-1}/2$ is satisfied only for $\lambda=0$, except perhaps at the initialization, since along direction $d_{n-1}$, $\Delta_{n-1}$ is the minimum displacement possible for $\lambda$. Moreover it is checked if $\Delta$ still satisfies the bounds imposed by the global step size $\Phi$, if this is not the case, it is updated to the corresponding upper (or lower) bound. 

\subsection{Proof of Theorem 1}

Denote as \textit{blocked points} all the points $x_{kj}$ such that
\begin{equation*}
f(x_{kj}\pm\Delta_{kj} d_{kj})>f(x_{kj}) - \rho(\Delta_{kj}), \quad \quad \forall j=0,1,...,n-1
\end{equation*}
namely points where no improvement is found along any of the exploring directions $d_j$.
\begin{lemma}\label{Lem}
	The sequence of step sizes $\{\Delta_{kji}\}$ produced by the line minimization procedure of Alg. 2 is such that $\Delta_{kji}\to 0$ as $k\to\infty$ for all $j\in\{0,1,...,n-1\}$ and $i\in[0,i_{kj}^\star]$.
\end{lemma}
\begin{proof}
	
	By construction, $\lambda_s \Phi_k\leq\Delta_{kji}\leq\lambda_t\Phi_k$ and $\{\Phi_k\}$ is a non-increasing sequence that reduces at blocked points (see lines 14-15 of Alg. 1). Hence, if blocked points occurred infinitely often, then we would have that, at blocked points, $\Phi_k\to 0$, and thus $\Delta_{kji}\to0$.
	
	By contradiction, if blocked points were not to occur infinitely often, then it means that there exists $w\in\mathbb{N}$ such that $\Phi_k=\Phi_w$ for all $k\geq w$. Thus, given $\lambda_s \Phi_w=\epsilon>0$, it follows that $\Delta_{kji}\geq \epsilon$ for all $k\geq w$, $j\in\{0,1,...,n-1\}$ and $i\in[0,i_{kj}^\star]$. Hence, for all $k\geq w$, there exists $j\in\{0,1,...,n-1\}$ and $i_{kj}^\star>0$ such that 
	\begin{align*}
	f(x_{kji}+\Delta_{kji}d_{kj})\leq f(x_{kji})-\rho(\Delta_{kji})\leq f(x_{kji})-\rho(\epsilon) \forall i\leq i_{kj}^\star.
	\end{align*}
	As such, $f(x_{kji})$ would decrease without a bound, contradicting (A2).
\end{proof}
\begin{lemma}\label{Lem3}
	For every $n\in\N_{0}$, the function $\rho:\R_{\geq 0}\to\R_{\geq 0}$ defined in \eqref{rho} is $o(\Delta^n)$ for $\Delta\to 0$ and, given $\theta\in(0,1)$ and $\Delta\in(0,1)$, the series $\sum_{n=0}^{\infty} (\theta^n\Delta)^{\left(\frac{1}{\theta^n\Delta}\right)}$ is convergent.
\end{lemma}
\begin{proof}
	Consider, without loss of generality, $\rho(\Delta)=\Delta^{\frac{1}{\Delta}}$.
	
	Let us show that $\Delta^{\frac{1}{\Delta}}$ is $o(\Delta)$. From the definition of little-o notation, we want to prove that
	\begin{equation}\label{Delta0}
	\lim_{\Delta\to 0} \frac{\Delta^{\frac{1}{\Delta}}}{\Delta} = \lim_{\Delta\to 0}\Delta^{\frac{1-\Delta}{\Delta}} = 0.
	\end{equation}
	To do so, notice that, for $\Delta\to0$, there exists $\bar{\Delta}\in(0,1)$ such that $\Delta^{\frac{1-\Delta}{\Delta}}<\bar{\Delta}^{\frac{1-\bar{\Delta}}{\Delta}}$ for $\Delta\in(0,\bar{\Delta})$. Moreover, defining $n(\Delta)=\min_{\bar{n}\in\N}\frac{1-\bar{\Delta}}{\Delta}<\bar{n}$, $\bar{\Delta}^{\frac{1-\bar{\Delta}}{\Delta}}<\bar{\Delta}^{n(\Delta)}$. Since $\bar{\Delta}\in(0,1)$ and, by definition, $\lim_{\Delta\to0}n(\Delta)=+\infty$, $\lim_{\Delta\to0}\bar{\Delta}^{n(\Delta)}=0$. Hence, as $\lim_{\Delta\to 0} \Delta^{\frac{1-\Delta}{\Delta}}\leq\lim_{\Delta\to0}\bar{\Delta}^{n(\Delta)}=0$ and $\Delta^{\frac{1-\Delta}{\Delta}}\geq 0$ for all $\Delta\in\R_{\geq 0}$, the limit \eqref{Delta0} is proved.
	
	From a similar reasoning it follows that for every $n>0$, $\Delta^{\frac{1}{\Delta}}$ is $o(\Delta^n)$. 
	
	Consider now the series $\sum_{n=0}^{\infty} (\theta^n\Delta)^{\left(\frac{1}{\theta^n\Delta}\right)}$, we can rewrite it as
	\begin{align}\label{Series}
	\sum_{n=0}^{\infty} (\theta^n\Delta)^{\left(\frac{1}{\theta^n\Delta}\right)} = \left( \sum_{n=0}^{\infty} \left(\left(\theta^\frac{1}{\Delta}\right)^n\right)^{\left(\frac{1}{\theta^n}\right)}\right) \left( \sum_{n=0}^{\infty} \left((\Delta)^{\left(\frac{1}{\Delta}\right)}\right)^{\left(\frac{1}{\theta^n}\right)}\right).
	\end{align}
	By assumption $\theta^{\frac{1}{\Delta}}\in(0,1)$.
	
	Define as $\bar{\theta}\in\R_{>0}$ the smallest real number such that $\theta\leq\bar{\theta}$ and $1/\bar{\theta}\in\N$. Then, as the series $\sum_{n=0}^{\infty} \left(\left(\theta^\frac{1}{\Delta}\right)^n\right)^{\left(\frac{1}{\theta^n}\right)}$ is bounded by $\sum_{n=0}^{\infty} \left(\left(\bar{\theta}^\frac{1}{\Delta}\right)^n\right)^{\left(\frac{1}{\bar{\theta}^n}\right)}$, that is bounded by a subseries of $\sum_{n=0}^{\infty} \left(\bar{\theta}^\frac{1}{\Delta}\right)^n$, it will converge.
	
	For the same reasoning, since $\sum_{n=0}^{\infty} \left((\Delta)^{\left(\frac{1}{\Delta}\right)}\right)^{\left(\frac{1}{\theta^n}\right)}$ can be bounded by a subseries of $\sum_{n=0}^{\infty} \left(\Delta^\frac{1}{\Delta}\right)^n$, and all the terms in \eqref{Series} are positive, the whole series converges. 
\end{proof}
\begin{theorem}\label{TH5}
	Every limit point $x$ of the sequence of blocked points generated by Alg. 2 satisfies $\nabla f(x)=0$.
\end{theorem}
\begin{proof}
	Denote as $\{\bar{x}_{k}\}$ the sequence of blocked points. Then
	\begin{equation*}
	f(\bar{x}_k+\Delta_{kj}d_{kj}) - f(\bar{x}_k) > -\rho(\Delta_{kj}) \quad \quad \forall j\in\{0,1,...,n-1\}.
	\end{equation*}

	Notice that, by $\det(\col(d_0^T,d_1^T,...,d_{n-1}^T))>\delta_{\text{det}}$, compactness of the sublevel sets of $f$ and the fact that the length of new directions, computed in line 29 of Alg. 1, is the distance between two explored points (and thus bounded by the diameter of the initial compact sublevel set), the norm of $d_{kj}$, for all $j=0,1,...,n-1$ and $k\geq0$, is upper bounded by $d_{max}:=\max_{j=0,1,...,n}\{d_{0j}, \textnormal{diam}(\mathcal{L}_f(x_o))\}$, as well as lower bounded. 
	The sequence $\{d_{kj}\}$ is thus bounded, and as such, considering any limit point $\bar{d}_j$, we can conclude that
	\begin{align*}
	\nabla f(\bar{x})^T \bar{d}_j = \lim_{\bar{x}_k\to x, \Delta_{kj}\to 0} \frac{f(\bar{x}_k+\Delta_{kj}d_{kj}) - f(\bar{x}_k)}{\Delta_{kj}}\geq\hspace*{1.3cm}\geq \lim_{\Delta_{kj}\to 0} -\frac{
		\rho(\Delta_{kj})}{\Delta_{kj}} = 0.
	\end{align*}
	Since this result is valid also for $-\bar{d}_j$, it follows that $\nabla f(x)^T\bar{d}_j=0$. Moreover, since $\{\bar{d}_0,\bar{d}_1,$ $...,\bar{d}_{n-1}\}$ span $\R^n$, we can conclude that $\nabla f(x)=0$.
\end{proof}
\begin{theorem}\label{TH6}
	Every limit point $x$ of the sequence of blocked points generated by Alg. 2 is a minimum.
\end{theorem}
\begin{proof}
	By assumption (A1) and Theorem \ref{TH5}, we only need to show that every limit point of the sequence of blocked points is not a maximum. As, by (A1), we are assuming that every maximum is an isolated point, it follows by definition that, considering a local maximum $\bar{x}\in\R^n$, there exists $\epsilon_m>0$ such that $\forall x\neq\bar{x}\in\R^n$ such that $\|x-\bar{x}\|\leq \epsilon_m$, $f(x)<f(\bar{x})$.
	
	We will prove the result by contradiction. Suppose there exists a subsequence of blocked points converging to $\bar{x}$. Denote it as $\{x_l\}$. Since for each $j=0,1,...,n-1$, $\Delta_{kji}>0$ and $\Delta_{kji}\to 0$ for $k\to\infty$, there exists $\bar{l}>0$ such that $\forall l\geq \bar{l}$, $\|x_l-\bar{x}\|< \epsilon_m$. 
	
	If every term of the sequence $\{x_l\}$ is such that $x_{l}\neq \bar{x}$, then, by the sufficient decrease condition and the definition of local maximum, it follows that $x_l\not\to \bar{x}$, since $f(\bar{x})>f(x_{l})$ for all $l\geq \bar{l}$, contradicting that such a sequence exists.
	
	So the only way for such a sequence to exist is if for some $\bar{\bar{l}}\geq\bar{l}$, $x_{l}=\bar{x}$ for all $l\geq\bar{\bar{l}}$.
	
	As $f$ is analytic at $\bar{x}$, there exists an even $m>0$ such that the $m-th$ derivative of $f$ with respect $x$ is different from zero and, being $\bar{x}$ a maximum, its norm is lower than zero. Denote it as $f^m(x)$. Then, considering the Taylor expansion of $f(\bar{x}+\Delta d)$ around $\bar{x}$, and noticing that $\Delta^{\frac{1}{\Delta}}$ is $o(\Delta^m)$ and $\|d\|$ is lower bounded, there exists a $\bar{\Delta}\in(0,1)$ such that for all $\Delta\in(0,\bar{\Delta})$
	\begin{equation*}
	f^m(x)\|d\|^m\Delta^m<-\rho(\Delta),
	\end{equation*}
	and thus there exists $\underline{l}\geq \bar{\bar{l}}$ such that $x_{\underline{l}}\neq \bar{x}$ and $f(x_{\underline{l}})<f(\bar{x})$.
	
	
	Thus every limit point of the sequence of blocked points cannot be a maximum, hence they will all be minima.
\end{proof}
We prove now Theorem 1, namely that
\begin{equation*}
\lim_{k\to\infty} \|x_{kji}\|_{\mathcal{A}^\star}=0 \quad \text{for all }j\in\{0,1,...,n-1\}\text{ and }i\in[0,i_{kj}^\star].
\end{equation*}
Denote, without loss of generality, the sequence $\{x_{kji}\}$ as $\{x_k\}$.
Notice that the subsequence of blocked points $\{x_{b_k}\}$ of $\{x_{k}\}$ converges to $\mathcal{A}^\star$, by Theorem \ref{TH6}. Suppose, by contradiction, that a subsequence $\{x_{p_k}\}$ of $\{x_{k}\}$ does converge to a point $\bar{x}\notin\mathcal{A}^\star$, with $\|\bar{x}\|_{\cA^\star}>\epsilon_p$, for some $\epsilon_p>0$.

By definition of converging sequence, there exists a $p^\star>0$ such that, for all $p\geq p^\star$, $\|x_{p_k}-\bar{x}\|< \epsilon_p$. Denote as $f_{\epsilon_p}:=\inf_{\{x:\|x-\bar{x}\|< \epsilon_p\}}f(x)$. Pick $\epsilon_b>0$ such that $\sup_{\{x:\|x\|_{\mathcal{A}^\star}< \epsilon_b\}}f(x)<f_{\epsilon_p}$ and notice that
$\|\bar{x}\|_{\mathcal{A}^\star}>\epsilon_p+\epsilon_b$.

Then there exists a $b^\star>0$ such that for all $b\geq b^\star$, $\|x_{b_k}\|_{\mathcal{A}^\star}< \epsilon_b$.

Pick $\chi^\star=\max(p^\star,b^\star)$ and define as $\bar{b}^\star\geq \chi^\star$ the smallest $k$ such that $x_{\bar{b}^\star}$ is a blocked point and $\bar{p}^\star \geq \chi^\star$ the smallest $k$ such that $x_{\bar{p}^\star}$ belongs to the sequence $x_{p_k}$. Then, clearly, since $f(x_k)$ is a non-increasing sequence (by the sufficient decrease condition), for $k\geq \bar{p}^\star$, no point in $\{x:\|x-\bar{x}\|\leq \epsilon_p\}$ can be selected, thus reaching a contradiction.

\subsection{Proof of Theorem \ref{Thm2}}
We first show that $\mathcal{H}_{cl}$ is nominally well-posed and all maximal solutions are complete.
\begin{lemma}\label{Lem1}
	\textit{Let assumptions (A0)-(A2) hold, and $\tau^\star>0$, $\delta_{det}>0$, $0<\lambda_s<1<\lambda_t$, $\mu\in(0,1/\lambda_t)$, $\theta\in(0,1)$ and $\gamma\in \R_{\geq1}$. Then the hybrid closed-loop system $\mathcal{H}_{cl}$ in \eqref{HCL} is nominally well-posed.}
\end{lemma}
\begin{proof}
	The set $C$ and $D$ are clearly closed. 
	
	
	Both $F$ and $K$ are continuous functions in $C$ and thus outer semicontinuous and locally bounded. Moreover, being both single-valued, they are also convex for every $(x,x_c)\in C$.
	
	The set-valued map $G(\cdot,f(\cdot))$ is composed by linear functions, apart for an instance of $\alpha^+$ where the norm operator is present, which is continuous in the set of definition, and an instance of $\Delta_{n-1}^+$ where the $max$ function is used, which is continuous as well. The map $G_{c\setminus\tau^\star}$ is thus piecewise continuous. As all the inequalities in the discrete dynamics are not strict, at at the points of discontinuity, it includes both left and right limit. It is thus outer semicontinous by definition.
	
	Since $G(\cdot,f(\cdot))$ is piecewise continuous, it is locally bounded by continuity.
	
	The hybrid closed-loop system $\mathcal{H}_{cl}$ thus satisfies the hybrid basic conditions (Assumption 6.5 in \cite{Teel}) and is nominally well-posed by Theorem 6.8 in \cite{Teel}.
\end{proof}
\begin{lemma}\label{Lem2} \textit{Let assumptions (A0)-(A2) hold, and $\tau^\star>0$, $\delta_{det}>0$, $0<\lambda_s<1<\lambda_t$, $\mu\in(0,1/\lambda_t)$, $\theta\in(0,1)$ and $\gamma\in \R_{\geq1}$. Then all maximal solutions to $\mathcal{H}_{cl}$ are complete.}
\end{lemma}
\begin{proof}
	We prove completeness of maximal solutions to $\mathcal{H}_{cl}$ by invoking Proposition 6.10 in \cite{Teel} on existence of solutions, and showing that no maximal solution jumps outside of $C\cup D$ or has finite escape time.
	
	We first show that the viability condition in Proposition 6.10 in \cite{Teel} holds for all $(\xi,x_c)\in C\setminus D$, namely that $F(\xi,x_c)\cap T_C(\xi,x_c)\neq \emptyset$, with $T_C:\R^{n+l}\x\mathcal{X}_c\to\R^{n+l}\x\mathcal{X}_c$ the Bouligand tangent cone of $C$ at $(\xi,x_c)$. Since $0\in T_C(\xi,x_c)$ always, the viability condition is readily satisfied for all the state variables apart from $\xi$ and $\tau$. As the projection onto the $\xi$-subspace of $C\setminus D$ is empty, the viability condition is satisfied also for the $\xi$ state variable. Regarding the timer $\tau$, define the projection of $C$ and $D$ onto the $\tau$-subspace as $C_\tau:=[0,\tau^\star]$ and $D_{\tau}:=[\tau^\star,\infty)$.
	As the set $C_\tau\setminus D_\tau=[0,\tau^\star)$ is open to the right, we only need to check the viability condition at $\tau=0$. Since at $\tau = 0$, $\dot \tau = 1$, the viability condition is satisfied also for $\tau$.
	
	
	Then, by Proposition 6.10 in \cite{Teel}, there exists a nontrivial solution from every initial condition in $\R^{n+l}\times\mathcal{X}_c$. Moreover, since $G(C\cup D)\subset C\cup D$, the solutions to $\mathcal{H}_{cl}$ or have finite time escape or are complete. Notice that for all solutions to $\mathcal{H}_{cl}$, $\zeta(t)$ does not have finite escape time by assumption. We show completeness by showing that all the other components of $(\xi,x_c)$ for all solutions to $\mathcal{H}_{cl}$ are bounded. Indeed, by condition (A2) and the update rule for the new directions \eqref{NewD}, for all initial conditions $(\xi(0,0),x_c(0,0))\in\R^{n+l}\times\mathcal{X}_c$, the state variables $d_j$, with $j=0,1,...,n$, are upper bounded in norm by
	\begin{align*}
	d_{max}:=\max_{j=0,1,...,n}\{\|d_j(0,0)\|, \textnormal{diam}(\mathcal{L}_f(\max\{f(x(0,0)),z(0,0)\}))\},
	\end{align*} 
	where, given $A\subset \R^n$, $\text{diam}(A):= \sup_{x,y\in A} \|x-y\|$. Moreover, as the determinant of the matrix composed by the set of directions is lower bounded by $\delta_{\text{det}}>0$, the directions $d_j$ are also lower bounded in norm. Denote the lower bound as $d_{\min}\in\R$. Then $\Delta_j$, $j=0,1,...,n$, are upper bounded by
	\begin{align*}
	\begin{array}{ll}
	\Delta_{\max}:=(1+\gamma)\max\left\{\max_{j=0,1,...,n}\Delta_j(0,0), \frac{\textnormal{diam}(\mathcal{L}_f(\max\{f(x(0,0)),z(0,0)\}))}{d_{\min}},\lambda_s\Phi(0,0)\right\}.
	\end{array}
	\end{align*}
	Based on the same reasoning, $\Phi(t,j)\leq \Phi(0,0)$, $|\lambda(t,j)|\leq d_{\max}\Delta_{\max}$, $\|\alpha(t,j)\|\leq n d_{\max}\Delta_{\max}$, $\bar{\alpha}(t,j)\leq n d_{\max}\Delta_{\max}$, $z(t,j)\leq \max\{z(0,0), f(x(0,0))\}$, $\|x(t,j)\|_{\cA^\star}\leq d_{\max}^2\Delta_{\max}^2 d$ for all $(t,j)\in \dom (\xi,x_c)$. Hence any state variable of $\mathcal{H}_{cl}$ is bounded, thus all the maximal solutions to $\mathcal{H}_{cl}$ are complete.
\end{proof}
In order to prove stability of $\cA$, define the Lyapunov function $V(\xi,x_c)=z-f(\cA^\star)$. We stress that, given assumption (A1), $f(\cA^\star)$ is a scalar. Since $V$ is $\mathcal{C}^1$, it is possible to bound the growth of $V$ along any maximal solution $\phi$ to $\mathcal{H}_{cl}$ as
\begin{align}
V(\phi(\bar{t},\bar{j}))-V(\phi(\underbar{$t$},\underbar{$j$}))\leq \int_{\underbar{$t$}}^{\bar{t}} \frac{d}{dt} V(\phi(t,j(t)))dt +\nonumber\sum_{j=\underbar{$j$}+1}^{\bar{j}} [V(\phi(t(j),j))-V(\phi(t(j),j-1))],
\end{align}
where
\begin{align}
&\frac{d}{dt} V(\phi(t,j(t)))=\dot z(t,j(t)) = 0, \label{dV}\\
&V(\phi(t(j),j))-V(\phi(t(j),j-1))=\begin{cases}
0 & x_c\in D_{2,5} \\
z(t(j),j)-z(t(j),j-1)\leq 0 & x_c\in D_{1,3,4}, \label{V+}
\end{cases}
\end{align}
where $t(j)$ and $j(t)$ denote respectively the least time $t$ and the least index $j$ such that $(t,j)\in\dom\phi$, $D_{2,5}:=D_2\cup D_5$ and $D_{1,3,4}:=D_1\cup D_3\cup D_4$.

The above conditions follow directly from the definition of $F_c$ and $G_c$. Indeed $z$ changes only during jumps, and in that case, for $x\notin \cA^\star$, it can decrease for $x_c\in D_{1,3,4}$ and remain unchanged for $x_c\in D_{2,5}$.
However the Lyapunov function $V$ is not strictly nonincreasing since there exist initial conditions for $z$ and $x$ such that $z(0,0)<f(x(0,0))$. However, after at most 3 timer-cycles, when $D_3$ is reached, $z$ gets updated to $f(x)$.

The above nonincreasing conditions on $V$ are thus only valid for $t \geq 3\tau^\star$ and $j\geq2$, where $(t,j)=(0,0)$ initially. As we show next, this does not hinder the stability of the set $\cA$ and convergence to the set $\cA_e$ for the hybrid system $\mathcal{H}_{cl}$.  

By the above discussion, compactness of $\mathcal{A}$, the conditions on $V$, and Theorem 7.6 in \cite{Inv}, for all $\varepsilon_1>0$ there exists a $\delta_1>0$ such that
\begin{align}\label{pr1}
\|(\xi(3\tau^\star,2),x_c(3\tau^\star,2))\|_{\mathcal{A}}<\delta_1\implies \|(\xi(t,j),x_c(t,j))\|_\mathcal{A}<\varepsilon_1\ \forall t+j\geq3\tau^\star + 2
\end{align} 
Noticing that the set $\mathcal{A}$ is invariant for $\mathcal{H}_{cl}$, it follows, by Lemma \ref{Lem1} and Corollary 4.8 in \cite{Hyb2}, that
\begin{align}\label{pr2}
\forall \varepsilon_2>0,\forall T>0, \exists \delta_2>0:\|(\xi(0,0),x_c(0,0))\|_\mathcal{A}&<\delta_2 \implies\nonumber\\ &\|(\xi(t,j),x_c(t,j))\|<\varepsilon_2\ \forall t+j<T
\end{align}
By choosing in \eqref{pr2} $T=3\tau^\star+2$ and $\varepsilon_2=\delta_1$, we see that the definition of uniform stability is recovered with $\varepsilon = \varepsilon_1$ and $\delta = \delta_2$.
Thus $\mathcal{A}$ is uniformly stable for $\mathcal{H}_{cl}$.

To show attractivity of $\mathcal{A}_e$ we invoke Theorem 4.7 in \cite{Inv}, setting $\mathcal{U}:=\mathcal{R}_{\geq3\tau^\star+2}(\R^{n+l}\times\mathcal{X}_c)$, namely the set of states that are reachable after $3\tau^\star+2$ (see Definition 6.15 in \cite{Teel}). Notice that $\mathcal{U}$ is forward invariant due to Lemma \ref{Lem2} and the definition of reachable set. By referring to the remark at the bottom of the proof of Theorem 4.7, we set $T=3\tau^\star$ and $J=2$, and defining $u_C$ and $u_D$ in the statement of Theorem 4.7 respectively as \eqref{dV} and \eqref{V+}, for some $r\in V(\R^{n+l}\times\mathcal{X}_c)$, the trajectories of $\mathcal{H}_{cl}$ approach the largest weakly invariant subset of
\begin{equation}
V^{-1}(r)\cap \mathcal{U}\cap[u_C^{-1}(0)\cup(u_D^{-1}(0)\cap G(u_D^{-1}(0)))].
\end{equation}
The Lyapunov function $V$ is constant along solutions to $\mathcal{H}_{cl}$ in $D_2$, $D_5$ and the set  $\mathcal{A}_e$. By $m^+=1$ in $g_2$ and by $q^+=0$ in $g_5$ we can conclude that neither $D_2$ nor $D_5$ are (weakly) invariant. Indeed $\mathcal{A}_e$ is actually the largest (weakly) invariant set contained in $\mathcal{U}$ where $V$ is constant along maximal solutions whose range is contained in $\mathcal{U}$.

\subsection{Proof of Theorem \ref{TH3}}

	We will first show that, for any $\bar{n}_s>0$, these class of algorithms can potentially remain stuck at every $x\in\R^n$. As such, by continuous differentiability of $f$, for every compact set $\mathcal{C}\subset\R^n$, there exists a maximum gradient norm $\nabla f_\mathcal{C}$. Consider, without loss of generality, a unique step size variable $\Delta>0$ and a single direction $d\in\R^n$. By the mean value theorem, it follows that, for all $x,y\in\mathcal{C}$, $|f(x)-f(w)|\leq \nabla f_\mathcal{C}\|x-w\|$, and, for $w=x-p\Delta d$, at iteration $k$ in the algorithm
\begin{equation}\label{MVT}
|f(x_k)-f(x_k-p_k\Delta_k d_k)|\leq \nabla f_\mathcal{C} \Delta_k \bar{d},
\end{equation}
where, by continuous differentiability of $f$, $\nabla f_\mathcal{C}<\infty$ for all compact $\mathcal{C}\subset\R^n$, and $\|d_k\|\leq \bar{d}>0$.

Given a noise bound $\bar{n}_s>0$, and remembering that $\Delta_k\to0$ for $k\to\infty$, there exists a $k^\star>0$ such that
\begin{equation}\label{Condition}
\nabla f_\mathcal{C} \Delta_{k^\star} \bar{d}+\rho(\Delta_{k^\star})<\nabla f_\mathcal{C} \Delta_{k^\star} \bar{d}\frac{1}{1-\theta}+\iota \Delta_{k^\star}^\star<\bar{n}_s,
\end{equation}
with $\iota>0$ and $\Delta_{k^\star}^\star>0$ the value of the series $\sum_{n=0}^{\infty} (\theta^n\Delta_{k^\star})^{\left(\frac{1}{\theta^n\Delta_{k^\star}}\right)}$, proved to be convergent in Lemma \ref{Lem3}. The term $1/(1-\theta)$ follows by noticing that given iteration $k_1$, after $k_2$ iterations of blocked points, then $\Delta_{k_1+k_2}=\theta^{k_2}\Delta_{k_1}$, and, as $k_2\to\infty$, if we sum all the terms, we have a geometric series. We defined the bound in this way, since we build a noise function by iteratively summing previous noise values to produce the new one. 

A noise signal defined to be $n_s(k)=0$ for $k<k^\star$ and $n_s(k) = \nabla f_\mathcal{C} \Delta_{k} \bar{d}+\rho(\Delta_{k})+n_s(k-1)$ for $k\geq k^\star$ will keep the algorithm stuck in $x=x_{k^\star}$ for all $k\geq k^\star$.

The reason is that the following relationship will always be satisfied
\begin{equation*}
f(x_k+p_k\Delta_kd_k) + n_s(k) \geq f(x_k) + n_s(k-1) -\rho(\Delta_{k}),
\end{equation*}
where $f(x_k+p_k\Delta_kd_k) + n_s(k)$ is the cost function measurement obtained at iteration $k$ and $f(x_k) + n_s(k-1)$ is the cost function measurement obtained at the previous iteration. Namely no improvement is ever found in any direction, since
\begin{align} \label{W1}
n_s(k) = \nabla f_\mathcal{C} \Delta_{k} \bar{d}+\rho(\Delta_k)+n_s(k-1) \geq f(x_k)- f(x_k+p_k\Delta_kd_k) -\rho(\Delta_k) + n_s(k-1).
\end{align}

Now notice that at iterations where
\begin{equation*}
f(x_k+p_k\Delta_kd_k)\geq f(x_k) -\rho(\Delta_{k}),
\end{equation*}
namely at iterations where no improvement would be found in case of no noise, the noise could act in order to mistakenly consider an improvement. Indeed in that case, with a noise of the form
\begin{equation}\label{W2}
n_s(k) = -\nabla f_\mathcal{C} \Delta_{k} \bar{d}-\rho(\Delta_{k})+n_s(k-1),
\end{equation}
for $k\geq k_1^\star\geq0$ and $n_s(k_1^\star)=0$, a wrong descent direction will be picked from everywhere in $\mathcal{C}$.

Alternating the noise values of \eqref{W1} and \eqref{W2}, by considering $n_s(k-1)=0$ when switching strategy, as long as $\Delta_k\leq \Delta_{k^\star}$, can steer the algorithm to every point in $\mathcal{C}$.

Consider now a compact set $\mathcal{C}_1\supset \mathcal{C}$ and denote the maximum gradient norm of $f$ on $\mathcal{C}_1$ as $\nabla f_{\mathcal{C}_1}$, where $\nabla f_{\mathcal{C}_1}\geq \nabla f_{\mathcal{C}}$.

Applying the noise \eqref{W1} in $\mathcal{C}$, it is possible to notice that there exists a $k^\star_1\geq k^\star$ such that for $k=k^\star_1$ condition \eqref{Condition} is satisfied for $\nabla f_{\mathcal{C}_1}$. 

Now, by switching between noise expressions \eqref{W1} and \eqref{W2}, guaranteeing that $\Delta_k\leq\Delta_{k_1^\star}$, makes it possible to steer the sequence of iterate everywhere in $\mathcal{C}_1$ and in particular outside $\mathcal{C}$.

It is thus clear that repeating this procedure iteratively can make the sequence of iterate leave any compact sub-level set of $f$.

\subsection{Proof of Theorem \ref{TH4}}

	Let $\epsilon_1>\epsilon_2>0$ and constants $0<\lambda_s<1<\lambda_t$ be given.
	
	Notice that, by the bounds on the state variables defined in the proof of Lemma \ref{Lem2}, it is always possible to choose $\delta>0$ to be the maximum radius of all the balls, one per state variable composing $(x,x_c)$, such that the maximum of the bounds reported in the proof of Lemma \ref{Lem2} is upper bounded by $\epsilon_1$. Namely pick $\delta$ such that for all initial conditions in $\delta\mathbb{B}(\cA)$, 
	\begin{align*}
		\max\limits_{\delta>0}\max\limits_{(\xi(0,0),x_c(0,0))\in\delta \mathbb{B}(\cA)}\{d_{\max}, \Delta_{\max}, &\Phi(0,0), d_{\max}\Delta_{\max},  n d_{\max}\Delta_{\max},  \\&\max\{z(0,0),f(x(0,0))\}-f(\cA^\star),   d_{\max}^2\Delta_{\max}^2 \} <\epsilon_1.
	\end{align*}
	Then pick $S_{\Phi}=(0,\delta]$ and notice that, by \eqref{NewPhi}, for all $\underline{\Phi}\in S_{\Phi}$, $(\xi(0,0),x_c(0,0))\in\delta\mathbb{B}(\cA)\implies (\xi(t,j),x_c(t,j))\in\epsilon_1\mathbb{B}(\cA)$ for all $(t,j)\in\dom (\xi,x_c)$.

	Denote as $\mathcal{L}_{\epsilon}$ the largest sublevel set of $f$ subset of the closure of $\min\{\epsilon_2,\delta\}\mathbb{B}(\cA^\star)$, and as $f_{\epsilon}$ the biggest value that $f$ achieves in $\mathcal{L}_{\epsilon}$. 

	
	Pick $\mathcal{B}:=\text{cl}\{x\in\R^n:x\in\epsilon_1\mathbb{B}(\cA) \textnormal{ and }x\notin\mathcal{L}_{\epsilon}\}$.
	
	By assumptions (A0)-(A2) and the fact that the set of directions $d_j$, with $j=0,1,...,n-1$, always span $\R^n$, it follows (from Theorem \ref{TH6} and the fact that for outside any neighborhood small enough of the local maxima and local minima, the norm of the gradient of $f$ is lower bounded away from zero) that for every compact set in $\R^n$ not containing a local minimum, there exists a $\bar{\Phi}>0$, such that for all $\Phi\in(0,\bar{\Phi})$, there exists at least one direction, that, rescaled by $\lambda_s\Phi$, produces a sufficient decrease of $f$ from every point in that compact set. 
%
	
	Since $\mathcal{B}$ is compact and does not contain a local minimum, it implies, by the above reasoning, that there exists $\bar{\Phi}>0$, such that for all $\Phi_{bound}\in(0,\bar{\Phi})$ at least one direction is a descent direction for $\Phi=\Phi_{bound}$, hence, after at most $n$ iterations, $z$ decreases. 
	
	Define now the Lyapunov candidate function $V(\xi,x_c)=z-f_\epsilon$, and notice that it satisfies \eqref{dV} and \eqref{V+}, after at most $3\tau^\star+2$, on $\epsilon_1\mathbb{B}(\cA)\setminus \min\{\epsilon_2,\delta\}\mathbb{B}(\cA)$. By Lemma \ref{Lem} and picking $\underline{\Phi}\in(0,\bar{\Phi})$, it follows that there exist $(T,J)\in\dom (\xi,x_c)$ such that for all $(t,j)\in\dom (\xi,x_c)$ such that $t+j\geq T+J$, $\Phi=\underline{\Phi}$ and $\Delta_{j}=\lambda_s\underline{\Phi}$. This implies that
	as long as $x\in\mathcal{B}$, after at most every $n$ iterations, $z$ decreases, and thus, by applying Theorem 4.7 in \cite{Inv}, $(\xi,x_c)\to\epsilon_2\mathbb{B}(\cA)$.
	To conclude the proof, choose $\underline{\Phi}\in(0,\min\{\delta,\bar{\Phi}\})$.


\subsection{Proof of Corollary \ref{Cor1}}

In order to give a Lyapunov characterization of \eqref{Bound}, we first extend the results of \cite{AP1} to the hybrid systems framework, and, by properly defining a Lyapunov function, show the necessity of the bound based on \eqref{Bound} in order to guarantee stability of $\mathcal{H}_{cl}$.

\begin{theorem}
	Consider a function $V:U\to\R$, with $U\subset \R^n$ an open neighborhood of the invariant set $\mathcal{A}\subset\R^n$ for the hybrid system $\mathcal{H}$, and assume that $\mathcal{H}$ is nominally well-posed and pre-forward complete from $U$. We assume the following conditions are satisfied.
	\begin{enumerate}
		\item[$\cdot$] Condition 1:  $V(\mathcal{A})=0$ and $\forall x\in U$: $\alpha(\|x\|_{\mathcal{A}})\leq V(x)\leq \beta(\|x\|_{\mathcal{A}})$.  The functions $\alpha(\cdot)$ and $\beta(\cdot)$ are class $\mathcal{K}$ functions.
		\item[$\cdot$] Condition 2: There exists $T>0$ and there exists an open set $U'\subset U$ which contains $\mathcal{A}$ such that, for each hybrid arc $x:\dom x\to\R^n$ solution to $\mathcal{H}$, there exists an increasing sequence of times $(t_k^\star, j_k^\star)\in\dom x\ (k\in\mathbb{N})$ such that $(t_{k+1}^\star +j_{k+1}^\star)-(t_{k}^\star+j_k^\star)\leq T\ (\forall k\in\mathbb{N})$, such that $\forall k\in\mathbb{N}$ and $\forall x(t_k^\star,j_k^\star)\in U'\setminus \{\mathcal{A}\}$:
		\begin{equation}
		V(x(t_{k+1}^\star,j_{k+1}^\star))-V(x(t_{k}^\star,j_k^\star))\leq 0
		\end{equation}
		Then the set $\mathcal{A}$ is stable.
	\end{enumerate}
\end{theorem}
\begin{proof}
	Consider the set $\mathcal{A}+\epsilon\mathbb{B}$, with $\epsilon>0$ small enough, such that $\mathcal{A}+\epsilon\mathbb{B}\subset U'$. Apply Corollary 4.8 in \cite{Goebel2006} given $\epsilon$ and $T$, to get $\epsilon'>0$, since $\mathcal{A}$ is an invariant set, such that for all $x(t_0,j_0)\in\mathcal{A}+\epsilon\mathbb{B}$, $\|x(t,j)\|_\mathcal{A}\leq \epsilon'$ for all $(t,j)\in [(t_0,j_0),(t_0+T',j_0+J')]\subset \dom x$, with $T'+J'\leq T$.
	
	Define $\delta':=\beta^{-1}\alpha(\epsilon')$ and consider the closed ball $\mathcal{A}+\delta'\mathbb{B}$. Apply now Corollary 4.8 again, to $\delta'$ and $T$, and get $\delta''$. For all $x(t_0,j_0)\in\mathcal{A}+\delta''\mathbb{B}$, $\|x(t,j)\|_\mathcal{A}\leq \delta'$ for all $(t,j)\in [(t_0,j_0),(t_0+T',j_0+J')]\subset \dom x$, with $T'+J'\leq T$. By Condition 2, there exists a $k_0\in\mathbb{N}$ such that $(t^\star_{k_0}+j_{k_0}^\star)-(t_0+j_0)\leq T$, implying that $\|x(t^\star_{k_0},j_{k}^\star)\|_\mathcal{A}\leq\delta'$.
	
	For $x(t^\star_{k_0},j_{k_0}^\star)\in\mathcal{A}+\delta'\mathbb{B}$, $\beta(\|x(t^\star_{k_0},j_{k_0}^\star)\|_\mathcal{A})\leq\alpha(\epsilon')$, and thus $V(x(t^\star_{k_0},j_{k_0}^\star))\leq\alpha(\epsilon')$.
	
	Since, by Condition 2, $V(x(t^\star_{k_0+1},j_{k_0+1}^\star))\leq V(x(t^\star_{k_0},j_{k_0}^\star))$, $V(x(t^\star_{k_0+1},j_{k_0+1}^\star))\leq \alpha(\epsilon')$. Since $\alpha(\|x(t^\star_{k_0+1},j_{k_0+1}^\star)\|_\mathcal{A})\leq V(x(t^\star_{k_0+1},j_{k_0+1}^\star))\leq \alpha(\epsilon')$, one obtains $\|x(t^\star_{k_0+1},j_{k_0+1}^\star)\|_\mathcal{A}\leq\epsilon'$.
	
	By the same argument, $\forall n\in\mathbb{N}$, $x(t^\star_{k_0+n},j_{k_0+n}^\star)\in\mathcal{A}+\epsilon'\mathbb{B}$. It follows, from Corollary 4.8, that $\forall (t,j)\geq (t_0,j_0)$, $\|x(t,j)\|_\mathcal{A}\leq\epsilon$. Thus uniform stability is proved.
\end{proof}
Now consider the Lyapunov function $V(x)=f(x)-f(\mathcal{A}^\star)$ for the case of lower bounded step size $\Phi$, with lower bound $\underline{\Phi}>0$, and the sequence of times $(t(j_k^\star),j_k^\star)$, where $t(j)$ is the biggest time $t$ such that $(t,j)\in\dom x$, $(t(j_0^\star),j_0^\star)$ is such that $j_0^\star\geq j_0+3$ and, given $(t(j_k^\star),j_k^\star)$, $(t(j_{k+1}^\star),j_{k+1}^\star)$ is computed as $(t(j_{k}^\star+1),j_{k}^\star+1)$ if $(\xi(t(j_{k}^\star+1),j_{k}^\star+1),x_c(t(j_{k}^\star+1),j_{k}^\star+1))\notin D_2$.
We claim that such $V$ satisfies Condition 1 and Condition 2.

\textit{Condition 1} Clearly $V(\mathcal{A}^\star) = 0$, moreover, since $f$ is lower bounded, with lower bound $f(\mathcal{A}^\star)$ and continuous, $V$ can be bounded by
\begin{align*}
\bar{\alpha}(\|x\|_{\mathcal{A}^\star})\leq V(x) \leq \bar{\beta}(\|x\|_{\cA^\star}),
\end{align*}
where
\begin{align*}
\bar{\alpha}(\|x\|_{\mathcal{A}^\star})&:=(-e^{-\|x\|_{\cA^\star}}+1) \inf_{\bar{x}\in\R^n:\|\bar{x}\|_{\cA^\star}\geq \|x\|_{\cA^\star}} (f(\bar{x})-f(\cA^\star))\\
\bar{\beta}(\|x\|_{\cA^\star})&:=\|x\|_{\cA^\star}\max_{\bar{x}\in\R^n:\|\bar{x}\|_{\cA^\star}\leq \|x\|_{\cA^\star}} \|\nabla f(\bar{x})\|
\end{align*}
\textit{Condition 2} We first show that $(t(j_{k}^\star)+j_{k}^\star)-(t(j_{k+1}^\star),j_{k+1}^\star)\leq T=3T'+3$, with $T'>0$ the period of the timer.

Notice that $D_2$ is defined for $q\in\{0,1\}$, $\dot q = 0$ always, and for $(\xi,x_c)\in D_2$ $q^+=q+1$. Since the jump rule is defined by the timer only, given $(\xi(t(j_k^\star),j_k^\star),x_c(t(j_k^\star),j_k^\star))\in D\setminus D_2$, there exists $\bar{j}\in\{1,2,3\}$ such that $(\xi(t(j_k^\star+\bar{j}),j_k^\star+\bar{j}),x_c(t(j_k^\star+\bar{j}),j_k^\star+\bar{j}))\in D\setminus D_2$ again.

At $(t(j_k^\star),j_k^\star)$, $(\xi(t(j_k^\star),j_k^\star),x(t(j_k^\star),j_k^\star))\in D$.

As a cycle in the algorithm results in the state $x(t(j),j)$ moving between $D_i$ in the following order 
\begin{equation}
D_5 \rightarrow D_1\rightarrow...\rightarrow D_1\rightarrow D_2 \rightarrow D_3 \rightarrow D_4\rightarrow ...\rightarrow D_4\rightarrow D_5, 
\end{equation}
we can notice the following:\\
If $x(t(j_k^\star),j_k^\star)\in D_1$, then $y(t(j_k^\star),j_k^\star)=f(x(t(j_k^\star),j_k^\star))+n_s(t(j_k^\star),j_k^\star)\leq z(t(j_k^\star),j_k^\star)-\rho(\Delta(t(j_k^\star),j_k^\star))\leq f(x(t(j_k^\star-1),j_k^\star-1))+\bar{n_s}-\rho(\Delta(t(j_k^\star),j_k^\star))$.\\
If $x(t(j_k^\star),j_k^\star)\in D_2$, then we do not consider the Lyapunov function there.\\
If $x(t(j_k^\star),j_k^\star)\in D_3$, then $f(x(t(j_k^\star),j_k^\star))= f(x(t(j_k^\star-2),j_k^\star-2))$.\\
If $x(t(j_k^\star),j_k^\star)\in D_4$, then $y(t(j_k^\star),j_k^\star)=f(x(t(j_k^\star),j_k^\star))+n_s(t(j_k^\star),j_k^\star)\leq z(t(j_k^\star),j_k^\star)-\rho(\Delta(t(j_k^\star),j_k^\star))\leq f(x(t(j_k^\star-1),j_k^\star-1))+\bar{n_s}-\rho(\Delta(t(j_k^\star),j_k^\star))$.\\
If $x(t(j_k^\star),j_k^\star)\in D_5$, then $f(x(t(j_k^\star),j_k^\star)) = f(x(t(j_k^\star-2),j_k^\star-2))$.

If we assume $\bar{n_s}=0$, then we can notice that $V(x(t(j_{k+1}^\star),j_{k+1}^\star))-V(x(t(j_{k}^\star),j_{k}^\star))\leq 0$ for all hybrid arcs $x$, for all $k\in \mathbb{N}$.
\begin{align}
V(x(t(j_{k+1}^\star),j_{k+1}^\star))&-V(x(t(j_{k}^\star),j_{k}^\star))= \nonumber\\
&\begin{cases}
f(x(t(j_{k+1}^\star),j_{k+1}^\star))-f(x(t(j_{k}^\star),j_{k}^\star) & \text{ if }x(t(j_{k+1}^\star),j_k^\star)\in D_{1,4}\\
0 & \text{ if }x(t(j_{k+1}^\star),j_k^\star)\in D_{3,5}
\end{cases}
\end{align}

In case we assume noise acting on the cost function measurement, $n_s(t,j)\leq \bar{n_s}$, then, for $x(t(j_{k+1}^\star),j_k^\star)\in D_{1,4}$,
\begin{align}
V(x(t(j_{k+1}^\star),j_{k+1}^\star))-V(x(t(j_{k}^\star),j_{k}^\star))=f(x(t(j_{k+1}^\star),j_{k+1}^\star))-f(x(t(j_{k}^\star),j_{k}^\star))= \\ \delta_f(t(j_{k+1}^\star),j_{k+1}^\star).
\end{align}
Hence $\delta_f(t(j_{k+1}^\star),j_{k+1}^\star)$ totally determines the sign of $V(x(t(j_{k+1}^\star),j_{k+1}^\star))-V(x(t(j_{k}^\star),j_{k}^\star))$. 

Since
\begin{align*}
&y(t(j_{k+1}^\star),j_{k+1}^\star) = f(x(t(j_{k+1}^\star),j_{k+1}^\star)) +n_s(t(j_{k+1}^\star),j_{k+1}^\star)\leq z(t(j_{k}^\star),j_{k}^\star) -\Delta(t(j_{k}^\star),j_{k}^\star)^{\frac{1}{\Delta(t(j_{k}^\star),j_{k}^\star)}}\\ 
&= f(x(t(j_{k}^\star),j_{k}^\star)) +n_s(t(j_{k}^\star),j_{k}^\star) -\Delta(t(j_{k}^\star),j_{k}^\star)^{\frac{1}{\Delta(t(j_{k}^\star),j_{k}^\star)}} \implies \\
&f(x(t(j_{k+1}^\star),j_{k+1}^\star)) \leq f(x(t(j_{k}^\star),j_{k}^\star)) +2\bar{n}_s -\Delta(t(j_{k}^\star),j_{k}^\star)^{\frac{1}{\Delta(t(j_{k}^\star),j_{k}^\star)}}\implies \\
&f(x(t(j_{k+1}^\star),j_{k+1}^\star))- f(x(t(j_{k}^\star),j_{k}^\star)) = \delta_f(t(j_{k+1}^\star),j_{k+1}^\star) \leq 2\bar{n}_s -\rho(\Delta(t(j_{k}^\star),j_{k}^\star)).
\end{align*}

If $2\bar{n}_s -\rho(\lambda_s\underline{\Phi})\leq0$, then $V(x(t(j_{k+1}^\star),j_{k+1}^\star))-V(x(t(j_{k}^\star),j_{k}^\star))\leq 0$ for all $k\in\mathbb{N}$. Indeed, semiglobal practical stability is preserved for all $n_s:\R\times\mathbb{N}\to\R$, with $n_s(t,j)\leq \bar{n}_s$ for all $(t,j)\in\R\times\mathbb{N}$ and such that
\begin{equation}
\bar{n}_s\leq \frac{\rho(\lambda_s\underline{\Phi})}{2}
\end{equation}

\bibliographystyle{ifacconf}
\bibliography{Arxiv_submission}	

\begin{thebibliography}{23}
\providecommand{\natexlab}[1]{#1}
\providecommand{\url}[1]{\texttt{#1}}
\providecommand{\urlprefix}{URL }
\expandafter\ifx\csname urlstyle\endcsname\relax
  \providecommand{\doi}[1]{doi:\discretionary{}{}{}#1}\else
  \providecommand{\doi}{doi:\discretionary{}{}{}\begingroup
  \urlstyle{rm}\Url}\fi

\bibitem[{Aeyels and Peuteman(1998)}]{AP1}
Aeyels and Peuteman (1998).
\newblock {A new asymptotic stability criterion for nonlinear time-variant
  differential equations}.

\bibitem[{Azuma et~al.(2012)Azuma, Sakar, and Pappas}]{Azuma}
Azuma, S.i., Sakar, M.S., and Pappas, G.J. (2012).
\newblock {Stochastic Source Seeking by Mobile Robots}.
\newblock \emph{IEEE Transactions on Automatic Control}, 57(9), 2308--2321.
\newblock \doi{10.1109/TAC.2012.2186927}.

\bibitem[{Bachmayer and Leonard(2002)}]{Bachmayer}
Bachmayer, R. and Leonard, N.E. (2002).
\newblock {Vehicle Networks for Gradient Descent in a Sampled Environment}.
\newblock In \emph{Proceedings of the 41 Conference on Decision and Control}.

\bibitem[{Burian et~al.(1996)Burian, Yoerger, Bradley, and Singh}]{Burian}
Burian, E., Yoerger, D., Bradley, A., and Singh, H. (1996).
\newblock {Gradient Search with Autonomous Underwater Vehicles Using Scalar
  Measurements}.
\newblock In \emph{Proceedings on Symposium on Autonomous Underwater Vehicle
  Technology}.

\bibitem[{Byatt et~al.(2004)Byatt, Coope, and Price}]{Byatt2004}
Byatt, D., Coope, I.D., and Price, C.J. (2004).
\newblock {Conjugate grids for unconstrained optimisation}.
\newblock \emph{Computational Optimization and Applications}, 29(1), 49--68.
\newblock \doi{10.1023/B:COAP.0000039488.79258.d7}.

\bibitem[{Cochran and Krstic(2009)}]{Cochran}
Cochran, J. and Krstic, M. (2009).
\newblock {Nonholonomic Source Seeking With Tuning of Angular Velocity}.
\newblock \emph{IEEE Transactions on Automatic Control}, 54(4), 717--731.

\bibitem[{Coope and Price(1999)}]{Coope1999}
Coope, I.D. and Price, C.J. (1999).
\newblock {A Direct Search Conjugate Directions Algorithm for Unconstrained
  Minimization}.
\newblock Technical Report November.

\bibitem[{Fletcher(2000)}]{Roger}
Fletcher, R. (2000).
\newblock \emph{{Practical Methods of Optimization}}.
\newblock Second edi edition.

\bibitem[{Garcia-Palomares and Rodriguez(2002)}]{Garcia2002}
Garcia-Palomares, U. and Rodriguez, J. (2002).
\newblock {New Sequential and Parallel Derivative-free Algorithms for
  Unconstrained Minimization}.
\newblock \emph{SIAM Journal on Optimization}, 13(1), 79--96.

\bibitem[{Goebel and Teel(2006{\natexlab{a}})}]{Goebel2006}
Goebel and Teel (2006{\natexlab{a}}).
\newblock {Solutions to hybrid inclusions and graphical convergence with
  stability theory applications}.

\bibitem[{Goebel and Teel(2006{\natexlab{b}})}]{Hyb2}
Goebel, R. and Teel, A.R. (2006{\natexlab{b}}).
\newblock {Solutions to hybrid inclusions via set and graphical convergence
  with stability theory applications}.
\newblock \emph{Automatica}, 42, 573--587.
\newblock \doi{10.1016/j.automatica.2005.12.019}.

\bibitem[{Goebel et~al.(2012)Goebel, Sanfelice, and Teel}]{Teel}
Goebel, R., Sanfelice, R.G., and Teel, A.R. (2012).
\newblock \emph{{Hybrid Dynamical Systems Modelling, Stability and
  Robustness}}.

\bibitem[{Kolda et~al.(2003)Kolda, Lewis, and Torczon}]{Kolda2003}
Kolda, T.G., Lewis, R.M., and Torczon, V. (2003).
\newblock {Optimization by Direct Search: New Perspectives on Some Classical
  and Modern Methods}.
\newblock \emph{SIAM Review}, 45(3), 385--482.
\newblock \doi{10.1137/S003614450242889}.

\bibitem[{Lewis et~al.(2000)Lewis, Torczon, Trosset, and William}]{Lewis2000}
Lewis, R.M., Torczon, V., Trosset, M.W., and William, C. (2000).
\newblock {Direct Search Methods : Then and Now}.
\newblock \emph{Journal of Compational and Applied Mathematics}, 124, 191--207.

\bibitem[{Lucidi and Sciandrone(2002)}]{Lucidi2002}
Lucidi, S. and Sciandrone, M. (2002).
\newblock {On The Global Convergence of Derivative-free Methods for
  Unconstrained Optimization}.
\newblock \emph{SIAM Journal on Optimization}, 13(1), 97--116.

\bibitem[{Mayhew et~al.(2007)Mayhew, Sanfelice, and Teel}]{Sanf}
Mayhew, C.G., Sanfelice, R.G., and Teel, A.R. (2007).
\newblock {Robust Source-Seeking Hybrid Controllers for Autonomous Vehicles}.
\newblock \emph{American Control Conference}.

\bibitem[{Mayhew et~al.(2008{\natexlab{a}})Mayhew, Sanfelice, and Teel}]{Sanf3}
Mayhew, C.G., Sanfelice, R.G., and Teel, A.R. (2008{\natexlab{a}}).
\newblock {Robust hybrid source-seeking algorithms based on directional
  derivatives and their approximations}.
\newblock \emph{Proceedings of the IEEE Conference on Decision and Control},
  (2), 1735--1740.
\newblock \doi{10.1109/CDC.2008.4739392}.

\bibitem[{Mayhew et~al.(2008{\natexlab{b}})Mayhew, Sanfelice, and Teel}]{Sanf2}
Mayhew, C.G., Sanfelice, R.G., and Teel, A.R. (2008{\natexlab{b}}).
\newblock {Robust source-seeking hybrid controllers for nonholonomic vehicles}.
\newblock \emph{Proceedings of the American Control Conference}, 2722--2727.
\newblock \doi{10.1109/ACC.2008.4586904}.

\bibitem[{Popovi and Teel(2004)}]{NSTeel}
Popovi, D. and Teel, A.R. (2004).
\newblock {Direct Search Methods for Nonsmooth Optimization}.
\newblock \emph{Conference on Decitions and Control}, 3173--3178.
\newblock \doi{10.1109/CDC.2004.1428960}.

\bibitem[{Powell(1964)}]{Powell1964}
Powell, M.J.D. (1964).
\newblock {An efficient method for finding the minimum of a function of several
  variables without calculating derivatives}.

\bibitem[{Sanfelice et~al.(2007)Sanfelice, Goebel, and Teel}]{Inv}
Sanfelice, R.G., Goebel, R., and Teel, A.R. (2007).
\newblock {Invariance principles for hybrid systems with connections to
  detectability and asymptotic stability}.
\newblock \emph{IEEE Transactions on Automatic Control}, 52(12), 2282--2297.
\newblock \doi{10.1109/TAC.2007.910684}.

\bibitem[{Shkel and Lumelsky(2001)}]{Dubin}
Shkel, A.M. and Lumelsky, V. (2001).
\newblock {Classification of the Dubins set}.
\newblock \emph{Robotics and Autonomous Systems}, 34, 179--202.

\bibitem[{Smith(1962)}]{Smith}
Smith, C. (1962).
\newblock \emph{{The Automatic Computation of Maximum Likelihood Estimates}}.

\end{thebibliography}





\end{document}